\documentclass[a4paper,11pt]{article}

\usepackage{fullpage}
\usepackage{times}
\usepackage{soul}
\usepackage{url}
\usepackage[utf8]{inputenc}
\usepackage[small]{caption}
\usepackage{graphicx}
\usepackage[dvipsnames]{xcolor}
\usepackage{amsmath}
\usepackage{booktabs}
\usepackage{algorithm}
\usepackage{enumitem}
\usepackage{dsfont}
\usepackage[capitalize]{cleveref}
\usepackage{tikz}
\usepackage{multirow}

\usepackage{amsthm}
\usepackage{amssymb}
\usepackage[noend]{algorithmic}

\usepackage{bbm}
\usepackage{physics}
\usepackage{thm-restate}
\usepackage{pgfplots}
\usepackage{mathtools}

\usetikzlibrary{calc}
\usetikzlibrary{plotmarks}
\pgfplotsset{compat=1.3}
\pgfplotsset{
    legend image code/.code={
        \draw[%
                /pgfplots/mesh=false,
                bar width=3pt,%
                bar shift=0pt,%
                mark repeat=2,%
                mark phase=2]
                plot coordinates {
                    (0cm,0cm)
                    (0.15cm,0cm)
                    (0.3cm,0cm)%
                };%
    },   
}

\usepackage{etoolbox}
\makeatletter
  \patchcmd\pgfplots@assign@default@tick@foraxis@compute{\advance\MAX by0.5\H}{\advance\MAX by1.5\H}{}{}
\makeatother

\newtheorem{theorem}{Theorem}[section]

\newtheorem{proposition}[theorem]{Proposition}
\newtheorem{lemma}[theorem]{Lemma}

\theoremstyle{definition}
\newtheorem{definition}[theorem]{Definition}
\newtheorem{example}[theorem]{Example}

\crefname{claim}{Claim}{Claims}

\crefname{figure}{Figure}{Figures}

\newcommand{\R}{\mathbb{R}}

\newcommand{\UW}{\textup{UW}}
\newcommand{\EW}{\textup{EW}}
\newcommand{\MUW}{\textup{MUW}}

\DeclareMathOperator*{\argmax}{arg\,max}

\usepackage{natbib}
\usepackage{authblk}

\allowdisplaybreaks

\title{\bf Optimizing the Envy Cycle Elimination Algorithm}

\author{Karen Frilya Celine}
\author{Warut Suksompong}

\affil{National University of Singapore}

\date{\vspace{-10mm}}

\begin{document}

\maketitle

\begin{abstract}
In the fair allocation of indivisible goods, a widely used notion of fairness is \emph{envy-freeness up to one good (EF1)}.
A classical way to compute an EF1 allocation is the \emph{envy cycle elimination (ECE) algorithm}, which iteratively assigns a good to an unenvied agent and, after each assignment, resolves any resulting envy cycle.
Although the ECE algorithm always produces an EF1 allocation, it leaves considerable freedom in choosing both the next good to allocate and the agent to receive it.
We investigate natural heuristics that exploit this flexibility to improve welfare guarantees.
For example, we show that if the heuristic jointly selects the good and the receiving agent maximizing the utility, the worst-case utilitarian welfare loss is significantly lower than that of the vanilla algorithm.
By contrast, restricting the heuristic to select only one of these two dimensions does not yield comparable improvements.
We also complement our theoretical results with empirical average-case analysis.
\end{abstract}

\section{Introduction}
\label{sec:intro}

The problem of fairly allocating limited resources arises throughout society and has been extensively studied in the vibrant area of \emph{fair division} \citep{Moulin03,Moulin19}.
An application that has received substantial attention recently is the allocation of \emph{indivisible goods}, such as books, clothing, electronics, and artwork \citep{DemkoHi88,BramsEdFi03,AmanatidisAzBi23}.
A prominent fairness notion is \emph{envy-freeness (EF)}, which stipulates that no agent strictly prefers another agent's bundle to her own, that is, the former agent does not \emph{envy} the latter agent.
Unfortunately, EF is not always satisfiable when allocating indivisible goods---for example, if there are two agents but only one valuable good.
In light of this, a common relaxation is \emph{envy-freeness up to one good (EF1)}, which requires that if an agent strictly prefers another agent's bundle to her own, this envy can be resolved by removing some good from the other agent's bundle.

A simple and elegant way of obtaining an EF1 allocation is via the \emph{envy cycle elimination (ECE)} algorithm introduced by \citet{LiptonMaMo04}.
Due to its wide applicability, the ECE algorithm constitutes one of the most fundamental methods for fairly allocating indivisible goods.\footnote{The work of \citet{LiptonMaMo04} received the ACM SIGecom Test of Time Award in 2025.}
Whenever there is an ``envy cycle'' among agents, the algorithm eliminates it by passing bundles upstream along the cycle, so that each agent in the cycle receives the bundle of an agent that she envied.
Repeating this procedure leads to an allocation with an acyclic ``envy graph'', which means that at least one agent is unenvied.
The algorithm then chooses an available good and allocates it to an unenvied agent---this ensures that even if any envy is introduced, the envy can be resolved by removing at most one good.
That is, the allocation at any stage of the algorithm, and therefore also the final allocation, is EF1.

The EF1 guarantee of the ECE algorithm holds for any choice of available good and unenvied agent in each iteration.
Thus, for a given instance, the ECE algorithm can produce different allocations depending on these choices, with some allocations yielding higher welfare than others.
This raises the question of whether there exist heuristics for selecting the good and agent in each round that produce allocations with better welfare guarantees than the vanilla ECE algorithm.
If so, which heuristics exhibit strong performance with respect to each welfare objective?
We shall consider both utilitarian welfare (the sum of all agents' utilities) and egalitarian welfare (the minimum utility across all agents).

In this paper, we propose several natural ECE heuristics and study their welfare guarantees in terms of their \emph{strong price of fairness}.
This notion was introduced by \citet{BeiLuMa21} as a more robust variant of the well-established \emph{price of fairness} notion \citep{BertsimasFaTr11,CaragiannisKaKa12}, which serves to quantify the welfare loss caused by requiring allocations to satisfy a specific (fairness) property.
The (weak) price of a property is defined as the worst-case ratio between the maximum welfare over all allocations and the \emph{maximum} welfare over allocations satisfying the given property.
In other words, it measures the worst-case welfare loss of an \emph{optimal} allocation satisfying a certain property.
However, when it comes to evaluating the welfare guaranteed by algorithms, this measure is too optimistic---for example, the welfare loss of an arbitrary ECE allocation can be far worse than that of an optimal ECE allocation.
On the other hand, the strong price of a property is defined as the worst-case ratio between the maximum welfare over all allocations and the \emph{minimum} welfare over allocations satisfying the given property.
Consequently, the strong price provides a welfare guarantee for \emph{any} allocation satisfying a particular property.

As we shall see, there is a non-trivial gap between the weak and strong prices of the ECE algorithm, both in terms of utilitarian and egalitarian welfare.
This means that for some instances, the best and worst ECE allocations admit a significant welfare difference.
We demonstrate that several of our proposed heuristics lead to lower strong utilitarian or egalitarian prices compared to the vanilla ECE algorithm, thereby guaranteeing better welfare in the worst case.
In addition, we complement our theoretical results with experimental analysis showing how these heuristics fare in the average case.

\subsection{Our results}
\label{subsec:results}

We consider instances with $n \geq 2$ agents and $m \geq 1$ goods, and assume that all agents have additive utilities.\footnote{While the ECE algorithm guarantees EF1 for the more general class of monotone utilities, we demonstrate in Appendix~\ref{subsec:monotone} that the welfare guarantee of the ECE algorithm cannot be further improved under monotone utilities.
In particular, we show that under monotone utilities, the weak utilitarian and egalitarian prices of ECE are both~$\infty$.
On the other hand, most of our results can be extended to subadditive utilities---see Appendix~\ref{subsec:subadditive} for details.}
A more detailed description of our model is given in \cref{sec:prelim}, where we also formally define the envy cycle elimination (ECE) algorithm (\cref{alg:ece}).

In \cref{sec:no_heuristic}, we analyze the weak and strong prices of ECE.
We show that the weak utilitarian price of the ECE algorithm is at most $n$, while its strong utilitarian price is exactly~$n^2$.
Similarly, we prove that the weak egalitarian price of ECE is at most $2n-1$, while its strong egalitarian price is $\infty$.
Hence, for both utilitarian and egalitarian welfare, there is a significant gap between the weak and strong prices of ECE, which motivates our study of heuristics.

We then propose several heuristics for the ECE algorithm and study their strong utilitarian and egalitarian prices.
We start with simple greedy heuristics, which are categorized based on whether they choose only the agents, only the goods, or both.
We first assume in \cref{subsec:agent_heuristic} that the ordering of goods is prespecified, and the heuristic can only decide which unenvied agent to allocate each good to.
In \cref{subsec:good_heuristic}, we address the opposite case where the heuristic can only choose which good to allocate, while the receiving unenvied agent is selected according to a prespecified ordering of agents.
The remaining scenario where  both the goods and agents are chosen by the heuristic is studied
in \cref{subsec:agent_good_heuristic}.
For each of the three cases, we investigate the natural greedy heuristic where an unenvied agent, an unallocated good, or both are chosen to maximize the utility respectively.
Furthermore, in \cref{sec:max_min_matching}, we examine a non-greedy heuristic based on \emph{max-min matching}, where the allocation of the first $n$ goods is chosen to maximize the egalitarian welfare.

The strong utilitarian and egalitarian prices of these heuristics are displayed in \cref{tab:summary}.
In particular, both the heuristic that chooses an unenvied agent maximizing the utility and the heuristic that chooses an unallocated good maximizing the utility have similar strong utilitarian and egalitarian prices to the vanilla ECE algorithm.
On the other hand, merely choosing an unenvied agent and an unallocated good maximizing the utility in the first round (and choosing the agents and goods arbitrarily in subsequent rounds) already reduces the strong utilitarian price to strictly less than\footnote{Note that the weak utilitarian price of round-robin is $n$~\citep{BeiLuMa21}. 
Hence, this heuristic has a comparable utilitarian welfare guarantee to the round-robin algorithm with optimal agent ordering and tie-breaking mechanism.} $n + \sqrt{n/2}$, compared to $n^2$ without any heuristic.
Choosing such a pair of unenvied agent and unallocated good in every round does not yield much further improvement in terms of the strong utilitarian price, which is at least $n$;
however, the strong egalitarian price when $n = 2$ reduces to at most $3$, as opposed to~$\infty$ when optimizing only the first round.
For $n \geq 3$, none of our greedy heuristics offers a finite strong egalitarian price, even for fixed $n$ and $m$.
On the other hand, the non-greedy heuristic based on max-min matching notably improves the strong egalitarian price from $\infty$ to $m-n+1$.

In \cref{sec:experiments}, we complement our theoretical results with experimental analysis.
We generate instances by drawing each agent's utility for each good from some distribution, and normalizing the sum of each agent's utilities for all goods to~$1$.
In line with the theoretical results, the heuristic that chooses a pair of unenvied agent and good maximizing the utility outperforms all other heuristics (including the vanilla ECE algorithm) on average, in terms of both utilitarian and egalitarian welfare, for most tested values of $n$ and $m$.

\begin{table}[t]
\centering
\caption{Summary of results}
\begin{tabular}{|l|c|c|}
    \hline
    \multirow{2}*{Heuristic} & \multicolumn{2}{|c|}{Strong price} \\ \cline{2-3}
    & Utilitarian & Egalitarian \\ \hline
    \hline
    No heuristic & $n^2$ & $\infty$ \\ \hline
    Agent with maximum utility & $n^2$ & $\infty$ \\ \hline
    Good with maximum utility & $n^2-n+1$ & $\infty$ \\ \hline
    \hangindent=0.4cm Agent \& good with maximum utility (first round only) & $\Theta(n)$ & $\infty$ \\ \hline
    \multirow{2}{*}{\hangindent=0.4cm Agent \& good with maximum utility} & \multirow{2}{*}{$\Theta(n)$} & $n = 2$: $2$ to $3$ \\
    & & $n\geq3$: $\infty$ \\ \hline
    Max-min matching & $n^2$ & $m-n+1$ \\ \hline
\end{tabular}
\label{tab:summary}
\end{table}

\subsection{Further related work}
\label{subsec:related}

While ECE heuristics have not been studied systematically to our knowledge, certain heuristics have been used to obtain allocations that satisfy desirable properties on top of EF1.
For example, an ECE heuristic that chooses an unenvied agent with an empty bundle whenever possible, and assigns to the chosen agent an available good that she values the most, guarantees both $1/2$-EFX~\citep[Theorem 3.1]{MarkakisSa23} as well as $1/2$-MMS~\citep[Theorem 6]{AmanatidisAzBi23}.\footnote{We refer to those papers for the definitions of EFX and MMS.}
Furthermore, if all agents have the same rankings over the goods, the ECE heuristic that allocates the most valuable good in every round guarantees both EFX \citetext{\citealp[Lemma 3.5]{BarmanKr20}; \citealp[Theorem 6.2]{PlautRoughgarden20}} and $2/3$-MMS \citep[pages 3--4]{BarmanKr20}.
Concurrently with our work, \citet{AmanatidisBiRe26} investigated ECE from the perspective of equilibria and proved that for certain versions of ECE, agents who play best responses are guaranteed approximate EF1 when there are two or three agents.

\citet{BhaskarSrVa21} showed that a variant of the ECE algorithm can be used to obtain EF1 allocations for \emph{chores} by resolving not just any envy cycle, but specifically the ``top-trading'' envy cycle.
Choosing which envy cycle to resolve can also be seen as a type of heuristic, which is a potential direction for future work.
\citet{FeldmanMaNa24} considered randomizing the ECE algorithm to obtain an ex-ante $1/2$-EF guarantee.
Their randomized heuristic determines which unallocated good to assign to which unenvied agent in the first $n$ rounds, and decides which envy cycle to resolve in subsequent rounds.

\section{Preliminaries}
\label{sec:prelim}

Let $N = \{1, \dots, n\}$ be a set of $n \ge 2$ agents and $M = \{g_1, \dots, g_m\}$ be a set of $m \ge 1$ goods.
Each agent $i \in N$ has a utility function $u_i: 2^M \to \R_{\geq 0}$, which is \emph{normalized},\footnote{This normalization is common in research on the price of fairness \citep{CaragiannisKaKa12,BeiLuMa21,CelineDzKo23}.} that is, $u_i(M) = 1$.
Furthermore, we assume that the utility function is \emph{additive},\footnote{Some inequalities in our upper bound proofs are in fact equalities when utilities are additive, but we leave them as inequalities so that they also hold even when utilities are subadditive.
See Appendix~\ref{subsec:subadditive} for further discussion.}
that is, for each subset $S \subseteq M$, $u_i(S) = \sum_{g \in S} u_i(g)$, where $u_i(g)$ is a shorthand for $u_i(\{g\})$.
In particular, $u_i(\emptyset) = 0$.
An instance $\mathcal{I}$ is a tuple $(N, M, (u_i)_{i \in N})$.

A \emph{partial allocation} of an instance $\mathcal{I}$ is a tuple $\mathcal{A} = (A_1, \dots, A_n)$ such that $A_i \subseteq M$ for each $i \in N$ and $A_1, \dots, A_n$ are pairwise disjoint.
We say that $A_i$ is the bundle belonging to agent $i$ and $u_i(A_i)$ is the utility of agent $i$.
If $A_1 \cup \dots \cup A_n = M$, then we call $\mathcal{A}$ a \emph{complete allocation}; otherwise, $\mathcal{A}$ is said to be an \emph{incomplete allocation}.
Unless specified otherwise, we use the term \emph{allocation} to refer to a complete allocation.
The \emph{utilitarian welfare} of a partial allocation $\mathcal{A}$, denoted by $\UW(\mathcal{A})$, is the sum of all agents' utilities, $\sum_{i \in N} u_i(A_i)$.
The \emph{egalitarian welfare} of a partial allocation $\mathcal{A}$, denoted by $\EW(\mathcal{A})$, is the minimum across all agents' utilities, $\min_{i \in N} u_i(A_i)$.

For a property~$P$, let $P(\mathcal{I})$ denote the set of all allocations of instance $\mathcal{I}$ satisfying $P$.
The \emph{weak utilitarian price of $P$ for instance $\mathcal{I}$} is defined as
$
    \flatfrac{\max_{\mathcal{A}} \UW(\mathcal{A})}{\max_{\mathcal{A} \in P(\mathcal{I})} \UW(\mathcal{A}).}
$
Then, the \emph{weak utilitarian price of $P$} is the supremum of the weak utilitarian price of $P$ for instance $\mathcal{I}$ taken over all instances.
On the other hand, the \emph{strong utilitarian price of a property $P$ for instance $\mathcal{I}$} is defined as
$
    \flatfrac{\max_{\mathcal{A}} \UW(\mathcal{A})}{\min_{\mathcal{A} \in P(\mathcal{I})} \UW(\mathcal{A}).}
$
Similarly, the \emph{strong utilitarian price of $P$} is the supremum of the strong utilitarian price of $P$ for instance $\mathcal{I}$ taken over all instances.
We define the \emph{weak} and \emph{strong egalitarian prices} analogously, by replacing $\UW$ with $\EW$.
For division by zero, we assume that $0/0 = 1$ and $x/0 = \infty$ if $x > 0$.
Following the convention in this line of work, we generally consider the (weak or strong) price in terms of the number of agents $n$.
However, we will consider the price in terms of both $n$ and $m$ when the dependence on $m$ is important.

We next define several fairness properties.
Given a partial allocation $\mathcal{A}$, an agent $i \in N$ is said to \emph{envy} another agent $j \in N$ if $u_i(A_i) < u_i(A_j)$.
A partial allocation $\mathcal{A}$ is \emph{envy-free up to one good (EF1)} if for every pair of agents $i, j \in N$, either $A_j = \emptyset$ or there is a good $g \in A_j$ such that $u_i(A_i) \geq u_i(A_j \setminus \{g\})$.

One way to obtain an EF1 allocation is via the \emph{round-robin algorithm}, which works as follows:
We fix an ordering of agents, for example, $1, 2, \dots, n$.
In each round, the agents take turns picking an available good with the highest utility according to the fixed order, breaking ties arbitrarily.
We repeat the rounds until all goods have been allocated---note that if $m$ is not divisible by $n$, then some agents will not get to pick in the last round.
An allocation that is obtained through an execution of the round-robin algorithm, with some ordering of agents, is called a \emph{round-robin allocation}.

Another way to obtain an EF1 allocation is via the \emph{envy cycle elimination (ECE)} algorithm, which is described as \cref{alg:ece}.
Note that the envy relationship between agents under a partial allocation can be represented by a directed graph where the vertices correspond to the agents and there is an edge from an agent~$i$ to another agent~$j$ if and only if agent~$i$ envies agent~$j$; we call this graph the \emph{envy graph}.
An \emph{envy cycle} then refers to a cycle in the envy graph, i.e., a sequence of agents $(i_1, i_2, \dots, i_k, i_1)$ for some $k\ge 2$ such that for each $t \in \{1, 2, \dots, k-1\}$, agent $i_t$ envies agent $i_{t+1}$, and agent $i_k$ envies agent $i_1$.
An allocation that can be produced by the ECE algorithm is called an \emph{ECE allocation}.

\begin{algorithm}
\caption{Envy Cycle Elimination (ECE) Algorithm \citep{LiptonMaMo04}}
\label{alg:ece}
\begin{algorithmic}
\WHILE{there is an unallocated good}
    \STATE Choose an unallocated good~$g$ and an unenvied agent~$i$.
    \STATE Allocate good $g$ to agent $i$.
    \WHILE{there is an envy cycle $(i_1, i_2, \dots, i_k, i_1)$}
        \STATE Let $(A_1, \dots, A_n)$ be the current (partial) allocation.
        \STATE Allocate bundle $A_{i_1}$ to agent $i_k$.
        \FOR{$t \in \{1, 2, \dots, k-1\}$}
            \STATE Allocate bundle $A_{i_{t+1}}$ to agent $i_t$.
        \ENDFOR
    \ENDWHILE
\ENDWHILE
\end{algorithmic}
\end{algorithm}

We call an iteration of the outer while loop a \emph{round} of the ECE algorithm.
Hence, the algorithm consists of $m$ rounds. 
Observe that each agent's utility never decreases as the algorithm progresses, a fact that we will use multiple times in this paper.

We remark that our definition of the ECE algorithm differs slightly from the original one of \citet{LiptonMaMo04}.
Specifically, in our version, there is an extra execution of the envy cycle elimination step after the allocation of the last good.
While this last step is unnecessary for guaranteeing EF1, by ensuring that the final allocation does not admit an envy cycle, it significantly improves the utilitarian welfare guarantee in the worst case---see \cref{ex:ec_strong_price} for details.

Furthermore, note that the allocation obtained via the ECE algorithm is EF1 regardless of the choice of unenvied agent and unallocated good in each round.
We call the function that makes this choice an \emph{ECE heuristic}, as formally defined next.

\begin{definition}[Envy Cycle Elimination (ECE) Heuristic]
\label{def:heuristic}
A \emph{heuristic} for the ECE algorithm is a function $\mathcal{F}$ that takes as input an instance $\mathcal{I} = (N, M, (u_i)_{i \in N})$ and an incomplete allocation $\mathcal{A}' = (A'_i)_{i \in N}$ of $\mathcal{I}$, and outputs a pair of unenvied agent $i \in N$ and unallocated good $g_j \in M$.

Let $\mathcal{I}$ be an instance with $m$ goods and $\mathcal{A}$ be a complete allocation of this instance produced by the ECE algorithm.
We say that allocation $\mathcal{A}$ is obtained with the ECE heuristic~$\mathcal{F}$ if it is obtained through an execution of the ECE algorithm such that at the start of each round $r \in \{1, \dots, m\}$, if the partial allocation is $\mathcal{A}^{r-1}$, then an unallocated good $g_j$ is allocated to an unenvied agent $i$, where $(i, g_j)$ is the output of the function~$\mathcal{F}$ when given the input $(\mathcal{I}, \mathcal{A}^{r-1})$.
\end{definition}

While the ECE algorithm can often produce several EF1 allocations depending on how the unallocated good and the unenvied agent are chosen, it cannot always produce \emph{all} EF1 allocations.
In fact, as we show in \Cref{ex:ef1_not_ece}, there exists an instance with a unique allocation that maximizes the utilitarian welfare among all EF1 allocations, such that this allocation cannot be returned by ECE.

\section{No heuristic}
\label{sec:no_heuristic}

In this section, we examine the weak and strong prices of the ECE algorithm with no heuristic.
Recall that the weak price of ECE represents the worst-case welfare loss of the \emph{best} ECE allocation, while the strong price of ECE represents the corresponding loss of the \emph{worst} ECE allocation.
Therefore, the gap between the weak and strong prices of ECE represents the welfare difference between the best and worst ECE allocations.
We begin with the strong prices.

\begin{restatable}{proposition}{propstrongprice}
\label{prop:strong_price}
The strong utilitarian price of ECE is $n^2$.
The strong egalitarian price of ECE is~$\infty$.
\end{restatable}

For the upper bound on the strong utilitarian price, we will use the following lemma.

\begin{lemma}
\label{lem:no_envy_cycle}
If a complete allocation of an instance does not admit an envy cycle, then its utilitarian welfare is at least~$1/n$.
\end{lemma}

\begin{proof}
Let $\mathcal{A} = (A_1, \dots, A_n)$ be a complete allocation with an acyclic envy graph.
There must exist an agent who does not envy any other agent; we call her agent~$i$.
Then, for any agent~$j \in N$, we have $u_i(A_i) \geq u_i(A_j)$.
Summing this up over all agents $j \in N$ gives
    $n\cdot u_i(A_i) \geq \sum_{j \in N} u_i(A_j) \geq u_i(M) = 1$,
or $u_i(A_i) \ge 1/n$.
Hence, the utilitarian welfare of~$\mathcal{A}$ is at least $1/n$.
\end{proof}

We now prove \cref{prop:strong_price}.

\begin{proof}[Proof of \cref{prop:strong_price}]
We begin by proving the upper bound for the strong utilitarian price.
Since the utility functions are normalized, each agent can have a total utility of at most $1$.
Hence, the maximum utilitarian welfare is at most $n$.
On the other hand, by \cref{lem:no_envy_cycle}, the utilitarian welfare of an ECE allocation is at least $1/n$.
Therefore, the strong utilitarian price is at most $n/(1/n) = n^2$.

We prove the lower bound for both the strong utilitarian and egalitarian prices using an instance with $n$ agents, $m = 2n-1$ goods, and the following utilities:
\begin{itemize}
    \item $u_n(g_j) = 1/n$ for each $j \in \{1, \dots, n\}$.
    \item For each $i \in N \setminus \{n\}$, $u_i(g_{n+i}) = 1$.
    \item $u_i(g_j) = 0$ for all other pairs $(i, j)$.
\end{itemize}

Consider the following ECE allocation:
In each round~$r \in \{1, \dots, n\}$, good~$g_r$ is allocated to agent~$r$.
There is no envy cycle since in the first $n-1$ rounds, only agent~$n$ envies other agents, and at the end of round $n$, agent~$n$ no longer envies any other agent.
Next, in each round $r \in \{{n+1}, \dots, 2n-1\}$, good~$g_r$ is allocated to agent~$r-n+1$.
Again, there is no envy cycle since for each $i \in N \setminus \{n\}$, agent~$i$ only envies agent~$i+1$, and agent~$n$ envies no other agent.
In the final allocation, each agent~$i \in N \setminus \{n\}$ gets a utility of $0$, while agent~$n$ gets a utility of $1/n$. 
Hence, the utilitarian welfare of the allocation is $1/n$, and the egalitarian welfare is $0$.

Now, observe that for every good, there is a unique agent who values it positively; giving each good to such an agent results in an allocation where every agent gets a utility of $1$.
Thus, the maximum utilitarian welfare of this instance is $n$ and the maximum egalitarian welfare is $1$.
It follows that the strong utilitarian price of ECE for this instance is $n/(1/n) = n^2$, and its strong egalitarian price is $1/0 = \infty$.
\end{proof}

We next show that the weak prices of ECE are significantly lower than the strong prices.

\begin{proposition}
\label{prop:weak_price}
The weak utilitarian price of ECE is at most $n$.
The weak egalitarian price of ECE is at most $2n-1$.
\end{proposition}

To prove \cref{prop:weak_price}, we need the following lemma, which we find interesting in its own right.

\begin{restatable}{lemma}{lemrrimpliesece}
\label{lem:rr_implies_ece}
Any round-robin allocation can be obtained with the ECE algorithm.
\end{restatable}

\begin{proof}
Consider any allocation obtained with the round-robin algorithm.
Without loss of generality, assume that the agent ordering in each round\footnote{Note that in other parts of this paper, we use the term ``round'' to refer to an iteration of the ECE algorithm. However, for this proof specifically, we shall only use the term ``round'' to refer to a round of the round-robin algorithm, which consists of $n$ turns (except the last round, which may consist of fewer than $n$ turns).} is $(1, 2, \dots, n)$, and that the goods are allocated in the order $g_1, g_2, \dots, g_m$. 
In the ECE algorithm, we allocate the goods to the same agents as in the round-robin algorithm, but in reverse order.
Let $k = m - {(\lceil m/n \rceil - 1)n}$ be the number of turns in the last (possibly incomplete) round of the round-robin algorithm.
Then, the goods $g_m, \dots, g_1$ are given to agents $$k, k-1, \dots, 1,\, n, n-1, \dots, 1,\, \dots,\, n, n-1, \dots, 1$$ respectively, in that order starting from $g_m$.
To prove that this is indeed a valid allocation produced by the ECE algorithm, it suffices to show that every intermediate allocation (including the final allocation) has an acyclic envy graph, and each good is allocated to an unenvied agent.

We begin my making some observations about the envy graph of any intermediate allocation in this execution of the ECE algorithm.
Note that the reverse order of the original round-robin ordering is itself a round-robin ordering, with the ordering $$(a_1, \dots, a_n) \coloneq (k, {k-1}, \dots, 1, n, {n-1}, \dots, {k+1})$$ in each round.
Furthermore, for each $i\in\{1,\dots,n\}$, the good allocated to agent~$a_i$ in round~$r$ of the reverse round-robin ordering is the $r$-th last good received by agent~$a_i$ in the original round-robin allocation.

Our first claim is that, in any intermediate allocation, if two agents $a_i, a_j \in N$ have the same number of goods and $i > j$, then agent~$a_i$ does not envy agent~$a_j$.
To see this, observe that in each round $r \in \{1, \dots, \lceil m/n \rceil\}$ of the reverse round-robin ordering, agent~$a_i$ receives her $r$-th good~$g_{m+1-(r-1)n-i}$ after agent~$a_j$ receives her $r$-th good~$g_{m+1-(r-1)n-j}$.
Since the ordering is reversed in the original round-robin algorithm, agent~$a_i$ must prefer her good~$g_{m+1-(r-1)n-i}$ to the good~$g_{m+1-(r-1)n-j}$ of agent~$a_j$.
Since both agents have the same number of goods, agent~$a_i$ does not envy agent~$a_j$. 

Next, we claim that if two agents have different numbers of goods, the agent with more goods does not envy the agent with fewer goods.
The reasoning is similar to that of the previous claim.
In particular, for each~$r$, agent~$a_i$ receives her $(r+1)$-th good after agent~$a_j$ receives her $r$-th good in the reverse round-robin ordering.
Since the ordering is reversed in the original round-robin algorithm, agent~$a_i$ must prefer her good to agent $a_j$'s good.
Therefore, agent~$a_i$ does not envy agent~$a_j$ overall.

Combining the two claims above, it follows that the envy graph is always acyclic.
Indeed, we can find a topological ordering of the envy graph as follows.
First, sort the agents in increasing order of their bundle size.
Then, among agents with bundles of the same size, sort the agents $a_i \in N$ in increasing order of their indices $i$.
By our two claims, no agent envies an agent placed earlier in the topological order.

Finally, since the reverse of the round-robin ordering is still a round-robin ordering, each good is allocated to an agent with the fewest goods, and among such agents, the agent~$a_i$ with the lowest index $i$.
It therefore follows from our two claims that the agent who receives the good must be unenvied before the allocation of the good.
\end{proof}

\cref{prop:weak_price} then follows from \cref{lem:rr_implies_ece} as well as the following known result.

\begin{lemma}
[\citealp{BeiLuMa21,CelineDzKo23}]
\label{lem:weak_price_rr}
The weak utilitarian price of round-robin is $n$. The weak egalitarian price of round-robin is at most $2n-1$.
\end{lemma}

As our results in this section demonstrate, there are non-trivial gaps between the weak and strong prices of ECE.
This means that despite the existence of an ECE allocation with relatively high welfare for every instance, the worst-case guarantees if one runs the ECE algorithm arbitrarily are much weaker.
In the next two sections, we examine heuristics for enhancing the ECE algorithm by exploiting its flexibility.

\section{Greedy heuristics}

In this section, we investigate simple greedy heuristics that always choose an available option yielding the highest utility in each round of ECE.
We categorize the heuristics based on whether they have the freedom to choose only the agents, only the goods, or both.

\subsection{Heuristic for agents}
\label{subsec:agent_heuristic}

We start by assuming that the heuristic has the liberty to choose only the agents, while the ordering of goods is fixed in advance.
Without loss of generality, suppose that in each round $r \in \{1, \dots, m\}$, good $g_r$ is allocated to an unenvied agent chosen by the heuristic.
In particular, the heuristic selects an unenvied agent who has the highest utility for the good among all unenvied agents.
Perhaps surprisingly, we show that this heuristic does not give better guarantees than the ECE algorithm with no heuristic.

\begin{restatable}{theorem}{thmagentwithmaxutil}
\label{thm:agent_with_max_util}
Consider the ECE algorithm using the heuristic where each specified good (chosen according to a predetermined ordering) is allocated to an unenvied agent with the highest utility for the good.
Its strong utilitarian price is~$n^2$, and its strong egalitarian price is $\infty$.
\end{restatable}

\begin{proof}
Let $0 < \epsilon < 1/(3n)$.
We construct an instance with $n$ agents, $m = 2n$ goods, and the following utilities:
\begin{center}
\setlength{\tabcolsep}{4pt}
\begin{tabular}{c|cccccccc}
    $g$      & $g_1$ & $g_2$ & $g_3$ & $\dots$ & $g_{n}$ \\
    \hline
    $u_1(g)$ & $\epsilon$ & $\epsilon$ & $0$ & $\dots$ & $0$ \\
    $u_2(g)$ & $\epsilon$ & $0$ & $\epsilon$ & $\dots$ & $0$ \\
    $\vdots$ & $\vdots$ & $\vdots$ & $\vdots$ & $\ddots$ & $\vdots$ \\
    $u_{n-1}(g)$ & $\epsilon$ & $0$ & $0$ &  $\dots$ & $\epsilon$ \\
    $u_n(g)$ & $2\epsilon$ & $\frac{1}{n}$ & $\frac{1}{n}-\epsilon$ & $\dots$ & $\frac{1}{n}-\epsilon$ \\
\end{tabular}
\qquad
\begin{tabular}{c|cccccccc}
    $g$      & $g_{n+1}$ & $g_{n+2}$ & $\dots$ & $g_{2n-1}$ & $g_{2n}$ \\
    \hline
    $u_1(g)$ & $0$ & $1-2\epsilon$ & $\dots$ & $0$ & $0$ \\
    $\vdots$ & $\vdots$ & $\vdots$ & $\ddots$ & $\vdots$ & $\vdots$ \\
    $u_{n-2}(g)$ & $0$ & $0$ & $\dots$ & $1-2\epsilon$ & $0$ \\
    $u_{n-1}(g)$ & $0$ & $0$ & $\dots$ & $0$ & $1-2\epsilon$ \\
    $u_n(g)$ & $\frac{1}{n}-3\epsilon$ & $0$ & $\dots$ & $0$ & $(n-1)\epsilon$ \\
\end{tabular}
\end{center}

Consider the allocation obtained by the ECE algorithm using the heuristic where in round $r$, good $g_r$ is allocated to an unenvied agent $i$ maximizing the utility $u_i(g_r)$.
Assume that any remaining ties are broken by choosing the agent with the lowest index.
In the first round, good $g_1$ must be allocated to agent $n$, which causes all other agents to envy agent $n$.
For each round $r \in \{2, \dots, n\}$, there is at least one agent with an empty bundle; such an agent will envy agent $n$.
Hence, good~$g_r$ must be allocated to agent $r-1$, who is the unenvied agent with the highest utility for it.
At the end of round~$r$, the direction of envy between agents $r-1$ and $n$ is reversed, so that agent $n$ envies agent $r-1$ instead; note that no envy cycle is created.
After the first $n$ rounds, agent $n$ is no longer envied by any agent but envies every other agent.

Subsequently, good $g_{n+1}$ must be allocated to agent $n$---the only unenvied agent---removing the envy from agent $n$ to every agent $i \in N \setminus \{1, n\}$.
However, agent $n$ still envies agent $1$, who therefore cannot receive a good in round $n+2$.
Hence, good $g_{n+2}$ must be allocated to agent $2$, the unenvied agent with the lowest index.
This causes agent $2$ to be envied by agent $1$.
Then, it can be shown by induction that for each round $r \in \{n+3, \dots, 2n\}$, good $g_r$ must be allocated to agent~$r-n$, the unenvied agent with the lowest index, causing her to be envied by agent $r-n-1$.
Furthermore, allocating good $g_{2n}$ to agent $n$ in round $2n$ removes agent $n$'s envy towards agent $1$, so no envy cycle is created.

In the final allocation, each agent $i \in N \setminus \{n\}$ receives a total utility of $\epsilon$, while agent $n$ receives a total utility of $1/n+(n-2)\epsilon$. 
The utilitarian welfare of this allocation is $1/n + (2n-3)\epsilon$, and the egalitarian welfare is $\epsilon$.

On the other hand, consider a different allocation where all the goods~$g_1, \dots, g_{n+1}$ are allocated to agent $n$, whereas the goods~$g_{n+2}, \dots, g_{2n}$ are allocated to agents $1, \dots, n-1$ respectively.
Note that agent $n$ receives a utility of $1-(n-1)\epsilon$ and each agent $i \in N \setminus \{n\}$ receives a utility of $1-2\epsilon$.
Hence, the utilitarian welfare is $n-(3n-3)\epsilon$ and the egalitarian welfare is at least $1-n\epsilon$.
Therefore, the strong utilitarian price for this instance is at least $(n-(3n-3)\epsilon)/(1/n+(2n-3)\epsilon)$, which approaches $n^2$ as $\epsilon \rightarrow 0$.
Furthermore, the strong egalitarian price for this instance is at least $(1-n\epsilon)/\epsilon$, which approaches $\infty$ as $\epsilon \rightarrow 0$.

The bound for the strong egalitarian price is obviously tight.
For the strong utilitarian price, the tightness of the bound follows from \cref{prop:strong_price}, since the set of allocations obtained with the ECE algorithm using any heuristic is a subset of all allocations produced by the ECE algorithm.
\end{proof}

\subsection{Heuristic for goods}
\label{subsec:good_heuristic}

We next consider the case where the heuristic decides only the goods to be allocated, while the receiving unenvied agents are chosen according to a predetermined order.
We consider a simple ordering of agent, where in each round~$r \in \{1, \dots, n\}$, agent~$r$ is allocated some good chosen by the heuristic (note that agent~$r$ is always unenvied since she has an empty bundle).
Furthermore, in each round~$r \in \{n+1, \dots, m\}$, an unenvied agent with the lowest index is allocated a good chosen by the heuristic.
As in \cref{subsec:agent_heuristic}, we consider a greedy heuristic that chooses an unallocated good with the highest utility for the specified agent.\footnote{As mentioned in \cref{subsec:related}, this heuristic is known to guarantee $1/2$-EFX~\citep[Theorem 3.1]{MarkakisSa23} and $1/2$-MMS~\citep[Theorem 6]{AmanatidisAzBi23}.}
We demonstrate that this heuristic yields a slightly lower strong utilitarian price than the vanilla ECE algorithm.
Note that both the lower and upper bound proofs hold even if an arbitrary unenvied agent is allocated an arbitrary good in round $n+1$ onward.

\begin{restatable}{theorem}{thmgoodwithmaxutil}
\label{thm:good_with_max_util}
Consider the ECE algorithm using the heuristic where each specified agent (chosen according to the ordering described above) is allocated an available good with the highest utility for the agent.
Its strong utilitarian price is $n^2-n+1$, and its strong egalitarian price is $\infty$.
\end{restatable}

\begin{proof}
We start with the upper bound for the strong utilitarian price.
Let $(A_1, \dots, A_n)$ be an allocation obtained with this heuristic; its utilitarian welfare is denoted by $\UW = \sum_{i \in N} u_i(A_i)$.
Moreover, let $(A^*_1, \dots, A^*_n)$ be an allocation maximizing the utilitarian welfare; its utilitarian welfare is denoted by $\MUW = \sum_{j \in N} u_j(A^*_j)$.
For each bundle $A_i$ in the ECE allocation, let $\MUW(A_i) = \sum_{j \in N} u_j(A^*_j \cap A_i)$ be the ``portion'' of the maximum utilitarian welfare contributed by bundle $A_i$.
Note that
\begin{align*}
    \MUW = \sum_{j \in N} u_j(A^*_j)
    \leq \sum_{j \in N} \sum_{i \in N} u_j(A^*_j \cap A_i)
    = \sum_{i \in N} \sum_{j \in N} u_j(A^*_j \cap A_i)
    &= \sum_{i \in N} \MUW(A_i).
\end{align*}
Let $k$ be an unenvied agent---such an agent exists by definition of the ECE algorithm.
That is, for any $j \in N$, $u_j(A_j) \geq u_j(A_k)$.
Then,
\begin{align*}
    \MUW(A_k) = \sum_{j \in N} u_j(A^*_j \cap A_k) 
    \leq \sum_{j \in N} u_j(A_k) 
    \leq \sum_{j \in N} u_j(A_j) 
    = \UW.
\end{align*}
Now, let $i \in N \setminus \{k\}$.
We consider three cases depending on the size of bundle $A_i$.

\underline{Case 1}: $|A_i| \geq 2$.
Let $g$ be the last good to be added into the bundle $A_i$ (due to the envy cycle elimination step, $A_i$ may belong to a different agent from agent $i$ when $g$ is added).
By definition of the heuristic, $g$ must be added into the bundle $A_i$ after round $n$.
Then, for any agent $j \in N$, it must hold that $u_j(A_j) \geq u_j(g_j) \geq u_j(g)$, where $g_j$ is the good allocated to agent $j$ in round $j$.
On the other hand, by definition of the ECE heuristic, the owner of the bundle $A_i \setminus \{g\}$ must be unenvied when $g$ is added into it.
Hence, we also have $u_j(A_j) \geq u_j(A_i \setminus \{g\})$ for any agent $j \in N$.
As a result, $u_j(A_i) \leq u_j(A_i \setminus \{g\}) + u_j(g) \leq 2u_j(A_j)$.
This implies that
\begin{align*}
    \MUW(A_i) = \sum_{j \in N} u_j(A^*_j \cap A_i) 
    \leq \sum_{j \in N} u_j(A_i) 
    \leq 2 \sum_{j \in N} u_j(A_j) 
    = 2\cdot\UW.
\end{align*}

\underline{Case 2}: $|A_i| = 1$.
Then, there is exactly one agent $j \in N$ such that $A^*_j \cap A_i \neq \emptyset$, which implies that
    $\MUW(A_i) = u_j(A^*_j \cap A_i) \leq 1$.

\underline{Case 3}: $|A_i| = 0$.
In this case, $\MUW(A_i) = 0$.

Let $N_1 = \{i \in N \setminus \{k\} : |A_i| = 1\}$ and $N_2 = \{i \in N \setminus \{k\} : |A_i| \geq 2\}$.
By our analysis of the three cases,
\begin{align*}
    \MUW \leq \sum_{i \in N} \MUW(A_i)
    &= \MUW(A_k) + \sum_{i \in N_2} \MUW(A_i) + \sum_{i \in N_1} \MUW(A_i) \\
    &\leq \UW + 2|N_2|\cdot\UW + |N_1|.
\end{align*}
Hence,
\begin{align*}
    \frac{\MUW}{\UW} \leq 1 + 2|N_2| + \frac{|N_1|}{\UW}
    &\leq 1 + 2|N_2| + n|N_1| \tag{since $\UW \geq 1/n$ by \cref{lem:no_envy_cycle}} \\
    &\leq 1 + 2|N_2| + n(n-1-|N_2|) \tag{since $N_1 \subseteq (N \setminus \{k\}) \setminus
    N_2$} \\
    &= n^2 - n + 1 + (2-n)|N_2| \\
    &\leq n^2 - n + 1. \tag{since $n \geq 2$ and $|N_2| \geq 0$}
\end{align*}
Therefore, the strong utilitarian price is at most $n^2-n+1$.

We continue with the lower bounds.
Consider an instance with $n$ agents, $m = n$ goods, and the following utilities:
\begin{itemize}
    \item $u_1(g_j) = 1/n$ for all $j$.
    \item $u_i(g_{i-1}) = 1$ for each $i \in N \setminus \{1\}$.
    \item $u_i(g_j) = 0$ for all other pairs $(i, j)$.
\end{itemize}

Consider the allocation obtained with the ECE algorithm using the heuristic where each specified agent is allocated a good that yields the highest utility for her.
Assume that any remaining ties are broken by choosing the good with the lowest index.
In the first round, agent~$1$ picks good~$g_1$ with utility $1/n$.
In the second round, all remaining goods have the same utility of~$0$ to agent~$2$, so she picks good~$g_2$ which has the lowest index.
Then, it can be shown by induction that in each round~$r \in \{3, \dots, n\}$, agent~$r$ values all remaining goods at~$0$ and picks good~$g_r$ which has the lowest index.
Note that there is no envy cycle in any round, since each agent~$i \in N \setminus \{1\}$ only envies agent~$i-1$, and agent~$1$ does not envy any agent. 
In the resulting allocation, agent~$1$ receives a utility of $1/n$, while every other agent receives a utility of $0$. 
Hence, the utilitarian welfare of the allocation is $1/n$, and the egalitarian welfare is $0$.

Now, consider a different allocation where each agent~$i \in N \setminus \{1\}$ receives good~$g_{i-1}$ with utility $1$, and agent~$1$ receives good~$g_n$ with utility $1/n$.
The utilitarian welfare of this allocation is $n-1+1/n$, and the egalitarian welfare is $1/n$.
Therefore, the strong utilitarian price is at least $(n-1+1/n)/(1/n) = n^2-n+1$, and the strong egalitarian price is at least $(1/n)/0 = \infty$.
The bound for the strong egalitarian price is obviously tight.
\end{proof}

\subsection{Heuristic for both agents and goods}
\label{subsec:agent_good_heuristic}

In a bid to find heuristics that provide improved welfare guarantees, we turn our attention to the case where the heuristics can choose both the goods and the agents.
First, we consider a simple heuristic that chooses a pair of good and agent maximizing the utility in the first round, and chooses arbitrary goods and agents in subsequent rounds.
Remarkably, this heuristic already yields a much better strong utilitarian price than the vanilla ECE algorithm.

\begin{restatable}{theorem}{thmagentgoodwithmaxutilfirstroundonly}
\label{thm:agent_good_with_max_util_first_round_only}
Consider the ECE algorithm using the heuristic where in the first round, a good $g \in M$ is allocated to an agent $i \in N$ such that the utility $u_i(g)$ is maximized.
\begin{enumerate}[label=(\alph*)]
    \item If $n = 2$, its strong utilitarian price is $8/3$.
    \item If $n \geq 3$, its strong utilitarian price is 
    $n + \flatfrac{2ns}{(2n+s(s+1))}$,
    where $s = \lceil \sqrt{2n} - 1/2 \rceil$.
    This is strictly between $n + \sqrt{n/2} - 1/4$ and $n + \sqrt{n/2} - 1/5$.
    \item Its strong egalitarian price is $\infty$.
\end{enumerate}
\end{restatable}

We begin by proving \cref{thm:agent_good_with_max_util_first_round_only}(c), which is restated in the following lemma.

\begin{lemma}
\label{lem:first_round_only_strong_egal_price}
Consider the ECE algorithm using the heuristic where in the first round, a good $g \in M$ is allocated to an agent $i \in N$ such that $u_i(g)$ is maximized.
Its strong egalitarian price is $\infty$.
\end{lemma}

\begin{proof}
Consider an instance with $n \geq 2$ agents and $m = n$ goods where $u_i(g_j) = 1$ if $j = i$ and $u_i(g_j) = 0$ otherwise.
The maximum egalitarian welfare is $1$, which is obtained when each agent $i \in N$ receives the good $g_i$.
On the other hand, consider an allocation obtained via the ECE algorithm where in the first round, $g_1$ is allocated to agent $1$, and in the second round, $g_2$ is also allocated to agent $1$.
This allocation can be obtained with the heuristic in \cref{thm:agent_good_with_max_util_first_round_only}, since $u_i(g_j)$ is indeed maximized when $i = j = 1$.
Furthermore, agent~$1$ remains unenvied at the end of the first round, and therefore can receive good $g_2$ in the second round.
Since the ECE algorithm can only add new goods to existing bundles and change the ownership of existing bundles, $g_1$ and $g_2$ will stay in the same bundle for the remainder of the algorithm.
This leaves only $n-2$~goods to be distributed among the remaining ${n-1}$~bundles.
Hence, at least one agent must end up with an empty bundle.
It follows that the egalitarian welfare of this allocation is $0$, and the strong egalitarian price of this instance is $1/0 = \infty$.
\end{proof}

To establish \cref{thm:agent_good_with_max_util_first_round_only}(a) and (b), we will use the next two lemmas, whose proofs are quite involved.

\begin{lemma}
\label{lem:first_round_only_strong_util_price}
Consider the ECE algorithm using the heuristic where in the first round, a good $g \in M$ is allocated to an agent $i \in N$ such that $u_i(g)$ is maximized.
Its strong utilitarian price is
\[
    n + \max_{x \in [1/n, 1]} \left(\frac{1-x}{x + \sum_{i=2}^{\lceil 1/x \rceil} \frac{1-(i-1)x}{n}}\right).
\]
\end{lemma}

\begin{proof}
We begin by proving the upper bound.
Let an instance with $n \geq 2$ agents and $m \geq 1$ goods be given.
Let $v = \max_{i \in N,\, g \in M} u_i(g)$ be the maximum utility of any agent for any good, and $\mathcal{A} = (A_1, \dots, A_n)$ be an allocation of this instance obtained by the ECE algorithm using the heuristic where in the first round, some good of utility $v$ is allocated to an agent who values it at $v$.
Without loss of generality, assume that the agents $1, \dots, n$ are sorted in descending order of utility, that is, $u_1(A_1) \geq \dots \geq u_n(A_n)$.

We first find an upper bound on the maximum utilitarian welfare of this instance in terms of $n$, $v$, and the utilitarian welfare of the allocation~$\mathcal{A}$.
Observe that for each agent~$i \in N$ who is envied by some other agent, there exists a good~$g'_i \in A_i$, depending only on~$i$, such that removing this good from agent~$i$'s bundle resolves all envy towards her,\footnote{In particular, we can let $g'_i$ be the last good to be added to the bundle $A_i$, which cannot be empty since agent $i$ is envied by some other agent.
By definition of the ECE algorithm, the owner of the bundle~$A_i \setminus \{g'_i\}$ must be unenvied at the beginning of the round when~$g'_i$ is added to the bundle.
Since all agents' utilities do not decrease as the algorithm progresses, in the final allocation, every agent must value her own bundle at least as much as $A_i \setminus \{g'_i\}$.} that is, $u_j(A_j) \geq u_j(A_i \setminus \{g'_i\})$ for all agent $j \in N$.
Let~$M'$ be the set of all such goods~$g'_i$ belonging to any agent $i$ who is envied by another agent---if there are multiple such goods~$g'_i$ for some agent $i$, we only include one of them in $M'$.
Note that $|M'| \leq n-1$, since at least one agent is unenvied in $\mathcal{A}$ due to the final envy cycle elimination step.
Furthermore, for any agents $i, j \in N$, we have $u_i(A_i) \geq u_i(A_j \setminus M')$.
Hence, for any agent $i \in N$, we have
\begin{align*}
    u_i(M \setminus M') &= u_i\left(\left(\bigcup_{j \in N} A_j\right) \setminus M'\right) \\ 
    &= u_i\left(\bigcup_{j \in N} \left(A_j \setminus M'\right)\right) \\
    &\leq \sum_{j \in N} u_i(A_j \setminus M') \\
    &\leq n\cdot u_i(A_i).
\end{align*}

Consider the agent with the highest utility in allocation~$\mathcal{A}$, namely, agent~$1$.
Note that some agent is allocated a good of utility $v$ in the first round, and her utility never decreases throughout the execution of the ECE algorithm.
Furthermore, by the proof of \cref{lem:no_envy_cycle}, at least one agent receives a utility of at least $1/n$.
Since agent $1$ has the highest utility among all agents, we have $u_1(A_1) \geq \max\{1/n, v\}$.

We can now find an upper bound for the maximum utilitarian welfare.
Let $\mathcal{A}^* = (A^*_1, \dots, A^*_n)$ be an allocation maximizing the utilitarian welfare.
Then, the maximum utilitarian welfare, denoted by $\MUW$, is
\begin{align*}
    &\MUW = \sum_{i \in N} u_i(A^*_i) \\ 
    &\quad\leq \sum_{i \in N} u_i(A^*_i \cap M') + \sum_{i \in N} u_i(A^*_i \setminus M') \\
    &\quad\leq |M'|v + \sum_{i \in N} u_i(M \setminus M') \tag{since $u_i(A^*_i \cap M') \leq |A^*_i \cap M'|v$ and $A^*_i \setminus M' \subseteq M \setminus M'$}\\
    &\quad\leq (n-1)v + u_1(M \setminus M') + \sum_{i \in N \setminus \{1\}} u_i(M \setminus M')  \tag{since $|M'| \leq n-1$} \\
    &\quad\leq (n-1)v + 1 + \sum_{i \in N \setminus \{1\}} n\cdot u_i(A_i) \tag{since $u_i(M \setminus M') \leq \min\{1, n\cdot u_i(A_i)\}$} \\
    &\quad= (n-1)v + 1 - n\cdot u_1(A_1) + \sum_{i \in N} n\cdot u_i(A_i) \\
    &\quad\leq (n-1)v + 1 - \max\{1, nv\} + \sum_{i \in N} n\cdot u_i(A_i) \tag{since $u_1(A_1) \geq \max\{1/n, v\}$} \\
    &\quad= n\cdot\sum_{i \in N} u_i(A_i) + \min\{(n-1)v, 1-v\} \\
    &\quad= n\cdot \UW + \min\{(n-1)v, 1-v\},
\end{align*}
where $\UW = \sum_{i \in N} u_i(A_i)$ is the utilitarian welfare of the allocation $\mathcal{A}$.
Since this holds for any allocation that can be output by the specified heuristic, the strong utilitarian price of this heuristic for the given instance is at most
\begin{equation}
\label{eq:first_round_only_muw_ub}
    \frac{\MUW}{\UW} \leq n + \frac{\min\{(n-1)v, 1-v\}}{\UW}.
\end{equation}

Next, we derive an implicit lower bound on the utilitarian welfare $\UW$ of the allocation $\mathcal{A}$, by finding a lower bound on the utility of each agent, $u_i(A_i)$ for $i \in N$.
To this end, we use the fact that the envy graph of this allocation is acyclic, as ensured by the last envy cycle elimination step of the ECE algorithm.
Hence, the envy graph admits a topological ordering $\pi(n), \pi(n-1), \dots, \pi(1)$, where $\pi$ is a permutation of $N$ and for any $i, j \in N$ with $i < j$, agent $\pi(i)$ does not envy agent $\pi(j)$.
Then, for any $i, j \in N$ with $i \leq j$, we have $u_{\pi(i)}(A_{\pi(i)}) \geq u_{\pi(i)}(A_{\pi(j)})$.
Furthermore, for the remaining pairs $i, j \in N$ where $i > j$, agent $\pi(i)$ is EF1 towards agent $\pi(j)$, so $u_{\pi(i)}(A_{\pi(i)}) \geq u_{\pi(i)}(A_{\pi(j)} \setminus \{g\})$ for some good $g \in M$.
On the other hand, it holds that $u_{\pi(i)}(g) \leq v$ by definition of $v$.
Therefore, we have $u_{\pi(i)}(A_{\pi(i)}) \geq u_{\pi(i)}(A_{\pi(j)} \setminus \{g\}) \geq u_{\pi(i)}(A_{\pi(j)}) - v$.
It follows that for any $i \in N$,
\begin{align*}
    1 &= u_{\pi(i)}\left(\bigcup_{j=1}^n A_{\pi(j)}\right) \tag{since the instance is normalized} \\
    &\leq \sum_{j=1}^n u_{\pi(i)}(A_{\pi(j)}) \\
    &= \sum_{j=1}^{i-1} u_{\pi(i)}(A_{\pi(j)}) + \sum_{j=i}^n u_{\pi(i)}(A_{\pi(j)}) \\
    &\leq \sum_{j=1}^{i-1} \left(u_{\pi(i)}(A_{\pi(i)}) + v\right) + \sum_{j=i}^n u_{\pi(i)}(A_{\pi(i)}) \\
    &= n\cdot u_{\pi(i)}(A_{\pi(i)}) + (i-1)v.
\end{align*}
Equivalently, $u_{\pi(i)}(A_{\pi(i)}) \geq (1-(i-1)v)/n$ for each $i \in N$.
Note that for each $i, j \in N$ with $i \geq j$, it holds that $(1-(i-1)v)/n \leq (1-(j-1)v)/n$.
Furthermore, for each $i \in N$, there exist $i$ indices $j$ such that $i \geq j$ and thus $u_{\pi(j)}(A_{\pi(j)}) \geq (1-(i-1)v)/n$.
Hence, for each $i \in N$, there are at least $i$ agents with utility at least $(1-(i-1)v)/n$.
Then, by the assumption that $u_1(A_1) \geq \dots \geq u_n(A_n)$, it follows that for each $i \in N$,
\begin{equation}
\label{eq:first_round_only_agent_utility_lb}
    u_i(A_i) \geq \frac{1-(i-1)v}{n}.
\end{equation}

We can now derive the required upper bound on the strong utilitarian price for the given instance.
We first consider the case where~$v \leq 1/n$.
From \eqref{eq:first_round_only_agent_utility_lb}, we have
\begin{align*}
    \UW = \sum_{i=1}^n u_i(A_i)
    \geq \sum_{i=1}^{n} \frac{1-(i-1)v}{n}
    &= 1 - \frac{(n-1)v}{2}.
\end{align*}
Then, from \eqref{eq:first_round_only_muw_ub}, the strong utilitarian price for the given instance is
\begin{align*}
    \frac{\MUW}{\UW} &\leq n + \frac{\min\{(n-1)v, 1-v\}}{\UW} \\
    &\leq n + \frac{(n-1)v}{1 - (n-1)v/2} \tag{since $v \leq 1/n$} \\
    &= n + \frac{2(n-1)}{2/v-(n-1)} \\
    &\leq n + \frac{2n-2}{2n-n+1} \tag{since $v \leq 1/n$} \\
    &= n + \frac{2n-2}{n+1}.
\end{align*}
On the other hand, note that when $x = 1/n$,
\begin{align*}
    \frac{1-x}{x + \sum_{i=2}^{\lceil 1/x \rceil} \frac{1-(i-1)x}{n}} &= \frac{1-\frac{1}{n}}{\frac{1}{n} + \sum_{i=2}^{n} \frac{1-(i-1)/n}{n}} \\
    &= \frac{n-1}{1 + (n-1) - (n-1)/2} \\
    &= \frac{2n-2}{n+1}.
\end{align*}
Therefore, we have
\begin{align*}
    \frac{\MUW}{\UW} \leq n + \max_{x \in [1/n, 1]} \left(\frac{1-x}{x + \sum_{i=2}^{\lceil 1/x \rceil} \frac{1-(i-1)x}{n}}\right),
\end{align*}
which gives the required upper bound on the strong utilitarian price for the specified instance when $v \leq 1/n$.

We turn our attention to the remaining case where $v > 1/n$.
In this case, $1/v < n$, so $\lceil 1/v \rceil \leq n$.
Then, we have
\begin{align*}
    \UW = \sum_{i=1}^n u_i(A_i)
    \geq \sum_{i=1}^{\lceil 1/v \rceil} u_i(A_i) 
    &\geq v + \sum_{i=2}^{\lceil 1/v \rceil} \frac{1-(i-1)v}{n}.
\end{align*}
Here, the first inequality holds since $\lceil 1/v \rceil \leq n$.
Moreover, the last inequality follows from \eqref{eq:first_round_only_agent_utility_lb} and ${u_1(A_1) \geq v}$, which holds by the assumptions that agent $1$ has the maximum utility and that some agent is allocated a good of utility $v$ in the first round.
Subsequently, \eqref{eq:first_round_only_muw_ub} implies that the strong utilitarian price of the specified heuristic for the given instance is at most
\begin{align*}
    \frac{\MUW}{\UW} \leq n + \frac{\min\{(n-1)v, 1-v\}}{\UW} &= n + \frac{1-v}{v + \sum_{i=2}^{\lceil 1/v \rceil} \frac{1-(i-1)v}{n}} \tag{since $v > 1/n$} \\
    &\leq n + \max_{x \in [1/n, 1]} \left(\frac{1-x}{x + \sum_{i=2}^{\lceil 1/x \rceil} \frac{1-(i-1)x}{n}}\right).
\end{align*}
For the last inequality, note that $v > 1/n$ by the case assumption and $v = \max_{i \in N,\, g \in M} u_i(g) \leq 1$ since the instance is normalized. This completes our proof for the upper bound on the strong utilitarian price for the specified heuristic. 

Next, we prove the lower bound of the strong utilitarian price.
To this end, we show that for any $n \geq 2$ and $x \in [1/n, 1]$, there exists an instance with $n$ agents such that the strong utilitarian price for the given instance is exactly
\[
    n + \frac{1-x}{x + \sum_{i=2}^{\lceil 1/x \rceil} \frac{1-(i-1)x}{n}}.
\]
This implies that the upper bound proven earlier is tight.

Let $n \geq 2$, $x \in [1/n, 1]$, $t = \lceil 1/x \rceil$, and $\delta = 1 - (t - 1)x$.
Note that $t \leq n$ and $0 \leq \delta \leq x$.
We construct an instance with $n$ agents, $m = tn + n - 1$ goods, and the following utilities:
Each of the goods $g_1, \dots, g_{tn}$ is valued positively by at most one agent.
\begin{itemize}
    \item For each $i \in \{1, \dots, t\}$, every agent $j \neq i$ values the goods $g_{(i-1)n+1}, \dots, g_{in}$ at $0$.
    \item For each $i \in \{2, \dots, t\}$, agent $i$ values each of the goods $g_{(i-1)n+1}, \dots, g_{in}$ at $(1-(i-1)x)/n$.
    \item Agent $1$ values each of the goods $g_{1}, \dots, g_{t-1}$ at $x$ and good $g_t$ at $\delta$.
\end{itemize}
For the remaining $n-1$ goods, the utilities are as follows:
\begin{itemize}
    \item Agent $1$ values goods $g_{tn+1}, \dots, g_m$ at $0$.
    \item For each $i \in \{2, \dots, t\}$, agent $i$ values each of the goods $g_{tn+1}, \dots, g_{tn+i-1}$ at $x$ and the remaining goods at $0$.
    \item For each $i \in \{t+1, \dots, n\}$, agent $i$ values good~$g_{tn+1}$ at $\delta$, values each of the goods $g_{tn+i-(t-1)}, \dots,\linebreak[1] g_{tn+i-1}$ at $x$, and values the remaining goods at $0$.
\end{itemize}
Observe that $\max_{i \in N,\, g \in M} u_i(g) = x$.

We consider the following allocation obtained with the ECE algorithm using the heuristic where in the first round, a good $g \in M$ is allocated to an agent $i \in N$ such that $u_i(g)$ is maximized.
In the first round, good~$g_1$ is allocated to agent~$1$ who values it at $x$.
Then, in round~$r \in \{2, \dots, n\}$, good~$g_r$ is allocated to agent~$r$.
This creates no envy, since agent $1$ values each of the goods $g_1, \dots, g_n$ no more than $x$, and each agent $i \in N \setminus \{1\}$ values goods $g_1, \dots, g_n$ at $0$.

We turn our attention to rounds ${(i-1)n} + 1$ to $in$ for each $i \in \{2, \dots, t\}$.
In round~$r = (i-1)n+1$, good $g_r$ is allocated to agent $i$ who values it at $(1-(i-1)x)/n$.
Then, in rounds $(i-1)n+2, \dots, in$, goods $g_{(i-1)n+2}, \dots, g_{in}$ are allocated to agents in $N \setminus \{i\}$ in such a way that each of these $n-1$ agents receives exactly one of these $n-1$ goods.
Again, this creates no envy, since agent $i$ is the only agent who values any of the goods $g_{(i-1)n+1}, \dots, g_{in}$ positively, and she is the first agent to receive any of these goods.
Hence, there is no envy between any two agents at the end of round $t n$.

Next, in each round $r \in \{t n + 1,\linebreak[1] \dots, t n + n - 1\}$, good $g_r$ is allocated to agent $r - t n$.
While this introduces some envy between agents, we show that there is no envy cycle at any point.
To this end, observe that for each agent $i \in N$, it holds that $u_i(g_j) = 0$ whenever $j \geq t n + i$.
Since there is no envy at the end of round $t n$, and afterwards each agent $k \in N \setminus \{n\}$ only receives the good $g_{t n + k}$, agent $k$ is not envied by any agent $i \leq k$.
It follows that there is no envy cycle in the last $n-1$ rounds of the ECE algorithm as well.

We now compute the utilitarian welfare of this allocation.
Agent $1$ receives a utility of $x$ and each agent $i \in \{2, \dots, t\}$ receives a utility of $(1-(i-1)x)/n$.
Any other agent receives a utility of $0$.
Hence, the utilitarian welfare is exactly
\begin{equation}
\label{eq:first_round_only_lb_uw}
    \UW = x + \sum_{i=2}^{t} \frac{1-(i-1)x}{n} = x + \sum_{i=2}^{\lceil 1/x \rceil} \frac{1-(i-1)x}{n}.
\end{equation}

On the other hand, the maximum utilitarian welfare is obtained when each agent $i \in \{1, \dots, t\}$ receives the good $g_j$ for each $j \in \{(i-1)n + 1, \dots, in\}$, and each good $g_j$ with $j \in \{t n + 1, \dots, t n + n - 1\}$ is allocated to agent $j - t n + 1$.
Then, agent $1$ receives a utility of $1$, each agent $i \in \{2, \dots, t\}$ receives a utility of $1 - (i-2)x$, and each agent $i \in \{t + 1, \dots, n\}$ receives a utility of $x$.
Hence, the maximum utilitarian welfare is
\begin{align*}
    \MUW &= 1 + \sum_{i=2}^{t} (1-(i-2)x) + (n-t)x \\
    &= 1 + \sum_{i=2}^{t} (1-(i-1)x) + (n-1)x \\
    &= 1 + n \cdot \sum_{i=2}^{t} \frac{1-(i-1)x}{n} + (n-1)x \\
    &= 1 + n \cdot \sum_{i=2}^{\lceil 1/x \rceil} \frac{1-(i-1)x}{n} + (n-1)x \\
    &= 1 + n\cdot\UW - nx + (n-1)x \tag{from \eqref{eq:first_round_only_lb_uw}} \\
    &= 1 + n\cdot\UW - x.
\end{align*}

Therefore, the strong utilitarian price of this instance is
\begin{align*}
    \frac{\MUW}{\UW} &= n + \frac{1-x}{x + \sum_{i=2}^{\lceil 1/x \rceil} \frac{1-(i-1)x}{n}},
\end{align*}
completing the proof.
\end{proof}

\begin{lemma}
\label{lem:first_round_only_strong_util_price_argmax}
Let $n \geq 2$ be an integer. For each $x \in [1/n, 1]$, define
\begin{align*}
    f(x) = \frac{1-x}{x + \sum_{i=2}^{\lceil 1/x \rceil} \frac{1-(i-1)x}{n}}.
\end{align*}
If $n = 2$, then $1/2 \in \argmax_x f$.
Otherwise, $1/\lceil \sqrt{2n} + 1/2 \rceil \in \argmax_x f$.
\end{lemma}

\begin{proof}
For each integer $t \in \{1, \dots, n\}$ and real number $x \in [1/n, 1]$, we define
\[
    \varphi(t, x) = \frac{1-x}{x + \sum_{i=2}^t \frac{1-(i-1)x}{n}}.
\]
Then, $f(x) = \varphi(\lceil 1/x \rceil, x)$.

Since $f$ itself is not easily differentiable, we shall determine the slope of $f$ using the partial derivative of $\varphi$ with respect to $x$:
\begin{align*}
    \pdv{\varphi(t, x)}{x} &= \pdv{x} \left(\frac{1-x}{x + \sum_{i=2}^t \frac{1-(i-1)x}{n}}\right) \\
    &= \pdv{x} \left(\frac{1-x}{\frac{t-1}{n} + \left(1 - \frac{(t-1)t}{2n}\right)x}\right) \tag{by arithmetic sum} \\
    &= \pdv{x} \left(\frac{2n(1-x)}{2t - 2 + \left(2n + t - t^2\right)x}\right) \\
    &= 2n \left(-\frac{1}{2t - 2 + \left(2n + t - t^2\right)x} 
    - \frac{(2n + t - t^2)(1-x)}{\big(2t - 2 + (2n + t - t^2)x\big)^2} \right)
    \\
    &= \frac{2n(t^2-3t-2n+2)}{\big(2t - 2 + (2n + t - t^2)x\big)^2} \\
    &= \frac{2n(t^2-3t-2n+2)}{\big(2nx + (t-1)(2-tx)\big)^2}.
\end{align*}
Note that $2n > 0$.
Furthermore, for any $t \in \{2, \dots, n\}$ and $x \in [1/t, 1/(t-1)]$, we have $2nx + (t-1)(2-tx) > 2-tx \geq 2-t/(t-1) \geq 0$, which gives $(2nx + (t-1)(2-tx))^2 > 0$.
Solving the quadratic equation $t^2-3t-2n+2 = 0$ gives exactly one positive root $t = (3+\sqrt{8n+1})/2$ and one negative root $t = (3-\sqrt{8n+1})/2$.
Let $\alpha = (3+\sqrt{8n+1})/2$ be the positive root.
We can conclude that for any $t \in \{2, \dots, n\}$ and $x \in [1/t, 1/(t-1)]$:
\begin{itemize}
    \item $\pdv*{\varphi(t, x)}{x} > 0$ if $t > \alpha$,
    \item $\pdv*{\varphi(t, x)}{x} = 0$ if $t = \alpha$, and 
    \item $\pdv*{\varphi(t, x)}{x} < 0$ if $t < \alpha$.
\end{itemize}

We explain how the slope of $\varphi$ translates to the slope of $f$.
Note that for any $t \in \{2, \dots, n\}$ and $x \in [1/t, 1/(t-1))$, we have $t = \lceil 1/x\rceil$, and therefore $f(x) = \varphi(t, x)$.
Hence, for any $t \in \{2, \dots, n\}$ and $x \in (1/t, 1/(t-1))$, we have $f'(x) = \pdv*{\varphi(t, x)}{x}$.
Furthermore, $f$ is continuous from the right at $x \in \{1/n, 1/(n-1),\linebreak[1] \dots, 1/2\}$.
For each $t \in \{1, \dots, n-1\}$, the following calculation shows that $f$ is continuous from the left at $1/t$:
\begin{align*}
    \lim_{x \to (1/t)^-} f(x) &= \lim_{x \to (1/t)^-} \varphi(t+1, x) \tag{since $\lceil 1/x \rceil = t+1$ when $1/(t+1) \leq x < 1/t$} \\
    &= \varphi\left(t+1, \frac{1}{t}\right) \\
    &= \frac{1-\frac{1}{t}}{\frac{1}{t} + \sum_{i=2}^{t+1} \frac{1-(i-1)/t}{n}} \\
    &= \frac{1-\frac{1}{t}}{\frac{1}{t} + \sum_{i=2}^{t} \frac{1-(i-1)/t}{n}} \tag{since $(1-(i-1)/t)/n = 0$ when $i = t+1$} \\
    &= \varphi\left(t, \frac{1}{t}\right) \\
    &= f\left(\frac{1}{t}\right).
\end{align*}
Combined with the earlier observations, we have shown that $f$ is continuous on $[1/n, 1]$.

Now, let $\beta = (1+\sqrt{8n+1})/2\linebreak[1] = \alpha - 1$.
We show that $f$ is non-decreasing at $x$ if $x < 1/\lceil \beta \rceil$ and is strictly decreasing at $x$ if $x > 1/\lceil \beta \rceil$.
Let $x \in (1/n, 1)$ be such that $x < 1/\lceil \beta \rceil$.
Then, $\lceil 1/x \rceil \geq 1/x > \lceil \beta \rceil$, so $\lceil 1/x \rceil \geq \lceil \beta \rceil + 1 = \lceil \alpha \rceil \geq \alpha$.
If $x \not\in \{1/(n-1), 1/{(n-2)},\linebreak[1] \dots, 1/2\}$, then clearly $f$ is non-decreasing at $x$ since $f'(x) = \pdv*{\varphi(\lceil 1/x \rceil, x)}{x} \geq 0$.
Otherwise, let $x = 1/t$ where $t \in \{2, \dots, n-1\}$.
Since $t \geq \alpha$,
\begin{align*}
    \lim_{z \to (1/t)^+} f'(z) = \lim_{z \to (1/t)^+} \pdv*{\varphi(t, z)}{z} \geq 0,
\end{align*}
and since $t+1 > \alpha$,
\begin{align*}
    \lim_{z \to (1/t)^-} f'(z) = \lim_{z \to (1/t)^-} \pdv*{\varphi(t+1, z)}{z} \geq 0.
\end{align*}
Hence, $f$ is non-decreasing at $x$ in this case.

On the other hand, let $x \in (1/n, 1)$ be such that $x > 1/\lceil \beta \rceil$.
In this case, $1/x < \lceil \beta \rceil$, and so $\lceil 1/x \rceil \leq \lceil \beta \rceil < \beta + 1 = \alpha$.
If $x \not\in \{1/(n-1), 1/(n-2),\linebreak[1] \dots, 1/2\}$, then $f$ is strictly decreasing at~$x$ since $f'(x) = \pdv*{\varphi(\lceil 1/x \rceil, x)}{x} < 0$.
Otherwise, suppose that $x = 1/t$ where $t \in \{2, \dots, n-1\}$.
Since $t < \alpha$,
\begin{align*}
    \lim_{z \to (1/t)^+} f'(z) = \lim_{z \to (1/t)^+} \pdv*{\varphi(t, z)}{z} < 0.
\end{align*}
Furthermore, since $t = 1/x < \lceil \beta \rceil$ and $t$ is an integer, we have $t + 1 \leq \lceil \beta \rceil < \alpha$.
Then,
\begin{align*}
    \lim_{z \to (1/t)^-} f'(z) = \lim_{z \to (1/t)^-} \pdv*{\varphi(t+1, z)}{z} < 0.
\end{align*}
Thus, $f$ is strictly decreasing at $x$ in this case.

We now observe that $\lceil \beta \rceil = \lceil \sqrt{2n} + 1/2 \rceil$ for all integers $n \geq 2$.
To see this, first note that $\beta > (\sqrt{8n}+1)/2 = \sqrt{2n} + 1/2$.
Hence, $\lceil \beta \rceil \geq \lceil \sqrt{2n} + 1/2 \rceil$.
To prove the inequality in the opposite direction, notice that $\sqrt{2n} + 1/2$ is never an integer.
Indeed, if $\sqrt{2n} = k - 1/2$ for some integer $k$, then $2n = k^2 - k + 1/4$, contradicting the fact that $2n$ must be an integer.
Let $t = \lceil \sqrt{2n} + 1/2 \rceil$.
Then, $t > \sqrt{2n} + 1/2 = (\sqrt{8n}+1)/2$, or equivalently, $(2t - 1)^2 > 8n$.
Since $t$ is an integer, so is $(2t - 1)^2$, and therefore $(2t - 1)^2 \geq 8n+1$.
Hence, $\lceil \sqrt{2n} + 1/2 \rceil = t \geq \beta$.

To complete the proof, observe that when $n = 2$, we have $\lceil \sqrt{2n} + 1/2 \rceil > n$, or equivalently, $1/n > 1/\lceil \sqrt{2n} + 1/2 \rceil$.
This means that $f(x)$ is strictly decreasing on all $x \in [1/n, 1]$.
That is, the maximum of $f$ is obtained at $1/n = 1/2$ when $n = 2$.
On the other hand, if $n \geq 3$, then
\begin{align*}
    \sqrt{2n} + \frac{1}{2} &\leq \sqrt{(n-1)n} + \frac{1}{2} \\
    &= \sqrt{\left(\sqrt{n^2-n} + \frac{1}{2}\right)^2} \\
    &= \sqrt{n^2 - n + \sqrt{n^2-n} + \frac{1}{4}} \\
    &= \sqrt{n^2 + \sqrt{n^2-n} - \sqrt{\left(n - \frac{1}{4}\right)^2}} \\
    &= \sqrt{n^2 + \sqrt{n^2-n} - \sqrt{n^2 - \frac{n}{2} +\frac{1}{16}}} \\
    &< \sqrt{n^2} \tag{since $n^2-n < n^2-n/2+1/16$}\\
    &= n.
\end{align*}
Hence, $1/\lceil \sqrt{2n} + 1/2 \rceil \in [1/n, 1]$.
It then follows from our earlier observations that $1/\lceil \sqrt{2n} + 1/2 \rceil \in \argmax_x f$.
\end{proof}

We are now ready to prove \cref{thm:agent_good_with_max_util_first_round_only}.

\begin{proof}[Proof of \cref{thm:agent_good_with_max_util_first_round_only}]
The strong egalitarian price has been shown in \cref{lem:first_round_only_strong_egal_price}.
For the strong utilitarian price, the closed-form expression can be obtained from \cref{lem:first_round_only_strong_util_price,lem:first_round_only_strong_util_price_argmax}.
Let 
\[
    f(x) = \frac{1-x}{x + \sum_{i=2}^{\lceil 1/x \rceil} \frac{1-(i-1)x}{n}}
\]
be defined for any $x \in [1/n, 1]$.

If $n = 2$, then the strong utilitarian price is
\begin{align*}
    2 + f\left(\frac{1}{2}\right) = 2 + \frac{1-\frac{1}{2}}{\frac{1}{2} + \sum_{i=2}^{2} \frac{1-(i-1)/2}{2}}
    = 2 + \frac{\frac{1}{2}}{\frac{1}{2} + \frac{1}{4}}
    &= \frac{8}{3},
\end{align*}
proving \cref{thm:agent_good_with_max_util_first_round_only}(a).

Suppose now that $n \geq 3$.
Let $s = \lceil \sqrt{2n} - 1/2 \rceil$ be defined as in \cref{thm:agent_good_with_max_util_first_round_only}(b).
From \cref{lem:first_round_only_strong_util_price,lem:first_round_only_strong_util_price_argmax}, the strong utilitarian price is
\begin{align*}
    n + f\left(\frac{1}{s+1}\right) &= n + \frac{1-\frac{1}{s+1}}{\frac{1}{s+1} + \sum_{i=2}^{s+1} \frac{1-(i-1)/(s+1)}{n}} \\
    &= n + \frac{ns}{n + s(s+1) - s(s+1)/2} \\
    &= n + \frac{2ns}{2n + s(s+1)},
\end{align*}
as claimed in \cref{thm:agent_good_with_max_util_first_round_only}(b).

It remains to give the simplified approximation of the strong utilitarian price when $n \geq 3$, that is, to prove that
\begin{equation}
\label{eq:first_round_only_simplified}
    n + \sqrt{\frac{n}{2}} - \frac{1}{4} < n + f\left(\frac{1}{s+1}\right) < n + \sqrt{\frac{n}{2}} - \frac{1}{5}.
\end{equation}

We start by proving the upper bound.
Define
\[
    h(t) = \frac{2nt}{2n + t(t+1)}
\]
for any $t \in \R_{>0}$.
Then, $f(1/(s+1)) = h(s)$.
We shall prove that $h(s) \leq h(\sqrt{2n})$ using the derivative of $h$:
\begin{align*}
    h'(t) = \frac{2n(2n+t+t^2)-2nt(1+2t)}{(2n+t+t^2)^2}
    &= \frac{2n(2n-t^2)}{(2n+t+t^2)^2}.
\end{align*}
Note that for any $t \in \R_{>0}$, we have both $2n > 0$ and $(2n+t+t^2)^2 > 0$, and $2n-t^2 \geq 0$ holds if and only if $t \leq \sqrt{2n}$.
Hence, the maximum of~$h$ is achieved at $t = \sqrt{2n}$, and therefore $h(s) \leq h(\sqrt{2n})$.
That is,
\begin{align*}
    f\left(\frac{1}{s+1}\right) = h(s) \leq h\left(\sqrt{2n}\right)
    &= \frac{2n\sqrt{2n}}{4n +\sqrt{2n}}.
\end{align*}
Let $y = \sqrt{2n}$. Then, we have
\begin{align*}
    \frac{2n\sqrt{2n}}{4n +\sqrt{2n}} = \frac{y^3}{2y^2 + y} 
    &= \frac{y}{2} - \frac{1}{4} + \frac{1}{4(2y + 1)} \\
    &= \sqrt{\frac{n}{2}} - \frac{1}{4} + \frac{1}{4(2\sqrt{2n} + 1)} \\
    &< \sqrt{\frac{n}{2}} - \frac{1}{4} + \frac{1}{20} \tag{since $n > 2$} \\
    &= \sqrt{\frac{n}{2}} - \frac{1}{5}.
\end{align*}
Hence, $f(1/(s+1)) < \sqrt{n/2} - 1/5$.

We next prove the lower bound of \eqref{eq:first_round_only_simplified}.
Let $v = 1/(\sqrt{2n} + 1/2)$.
Recall from the proof of \cref{lem:first_round_only_strong_util_price_argmax} that $\sqrt{2n} + 1/2$ is never an integer, that is, $\sqrt{2n} + 1/2 < \lceil \sqrt{2n} + 1/2 \rceil = s+1$.
This means that $v > 1/(s+1) \geq 1/n$, where the last inequality holds since $n\ge 3$.
Furthermore, we have shown in the proof of \cref{lem:first_round_only_strong_util_price_argmax} that $f$ is strictly decreasing at any $x \in [1/n, 1]$ with $x > 1/(s+1) = 1/\lceil\beta\rceil$.
Hence, $f(1/(s+1)) > f(v)$.
Therefore, it suffices to show that $f(v) \geq \sqrt{n/2} - 1/4$.
We have
\begin{align*}
    f(v) &= \frac{1-v}{v + \sum_{i=2}^{\lceil 1/v \rceil} \frac{1-(i-1)v}{n}} \\
    &= \frac{1-v}{v + \frac{\lceil 1/v \rceil - 1}{n} - \frac{(\lceil 1/v \rceil - 1)\lceil 1/v \rceil v}{2n}} \\
    &= \frac{2n(1-v)}{2nv + 2s - s(s+1)v} \tag{since $\lceil 1/v \rceil = s + 1$} \\
    &= \frac{2n(1-v)}{-v(s^2 + (1-2/v)s - 2n)} \\
    &= \frac{2n(1-v)}{-v[(s + (v-2)/(2v))^2 - ((v-2)/(2v))^2 - 2n]} \\
    &\geq \frac{2n(1-v)}{-v[-((v-2)/(2v))^2-2n]} \tag{since $(s + (v-2)/(2v))^2 \geq 0$} \\
    &= \frac{8n(1/v-1)}{(1-2/v)^2+8n}.
\end{align*}
Let $z = \sqrt{8n} = 2\sqrt{2n}$.
Then, $v = 1/(\sqrt{2n} + 1/2) = 2/(z+1)$.
It follows that
\begin{align*}
    f(v) \geq \frac{z^2((z+1)/2 - 1)}{(1-(z+1))^2+z^2}
    = \frac{z^2(z-1)}{4z^2}
    &= \sqrt{\frac{n}{2}} - \frac{1}{4}.
\end{align*}
Therefore, $f(1/(s+1)) > \sqrt{n/2} - 1/4$, as desired.
\end{proof}

We next consider the heuristic that chooses a pair of unenvied agent and unallocated good maximizing the utility in \emph{every round}.
While there is no significant improvement in the strong utilitarian price compared to the heuristic that chooses such a pair of agent and good \emph{only in the first round}, the strong egalitarian price when $n = 2$ improves from $\infty$ to at most $3$.

\begin{restatable}{theorem}{thmagentgoodwithmaxutil}
\label{thm:agent_good_with_max_util}
Consider the ECE algorithm using the heuristic where in every round, an available good~$g$ is allocated to an unenvied agent $i$ such that the utility $u_i(g)$ is maximized.
\begin{enumerate}[label=(\alph*)]
    \item Its strong utilitarian price is in $[n, n + \sqrt{n/2} - 1/5]$.
    \item If $n = 2$, its strong egalitarian price is in $[2, 3]$.
    \item If $n \geq 3$, its strong egalitarian price is $\infty$.
\end{enumerate}
\end{restatable}

\begin{proof}
We begin by addressing the strong utilitarian price (\cref{thm:agent_good_with_max_util}(a)).
The upper bound follows from the strong utilitarian price of the ECE algorithm using the heuristic that chooses an unenvied agent and an available good in the first round (\cref{thm:agent_good_with_max_util_first_round_only}(a) and (b)),
since the set of allocations produced by the heuristic that chooses a pair of unenvied agent and available good maximizing the utility in every round is a subset of those obtained with the heuristic choosing such a pair in the first round.

For the lower bound, we construct an instance with $n$ agents and $m = x^n$ goods, where $x$ is a positive integer divisible by $n$.
Let $0 < \epsilon < 1/x^n$.
The agents' utilities are defined as follows:
\begin{itemize}
    \item For each $i \in N \setminus \{n\}$, $u_i(g_j) = (1-\epsilon)/x^i$ if $j \leq x^i$ and $u_i(g_j) = \epsilon/(x^n-x^i)$ otherwise.
    \item $u_n(g_j) = 1/x^n$ for each $j$.
\end{itemize}

We consider the allocation obtained by the ECE algorithm with the heuristic choosing a pair of unenvied agent and available good maximizing the utility in every round.
We assume that ties are broken by choosing an agent and a good with the lowest indices.
Observe that the goods $g_1, \dots, g_m$ are sorted in non-increasing order of utility for every agent.
Therefore, the goods must be allocated in the order $g_1, \dots, g_m$ to some unenvied agent who values it most.

In order to determine the utilitarian welfare of the resulting allocation, we show by induction that for each $t \in \{0, \dots, m/n-1\}$, each agent must receive exactly one of goods $g_{tn+1}, \dots, g_{tn+n}$.
This is obvious for the first $n$ rounds, since any agent with a non-empty bundle must be envied by another agent with an empty bundle.
Hence, the ECE algorithm will allocate the first $n$ goods to $n$ distinct agents.
We next establish the inductive step.
Let $t \in \{1, \dots, m/n-1\}$.
Assume, as the inductive hypothesis, that for each $t' \in \{0, \dots, t-1\}$, each agent receives exactly one of goods $g_{t'n+1}, \dots, g_{t'n+n}$.
Note that for every $t'' \in \{0, \dots, m/n-1\}$ and each agent $i \in N$, all of the goods $g_{t''n+1}, \dots, g_{t''n+n}$ have the same utility to agent~$i$.
Hence, the inductive hypothesis implies that in the partial allocation obtained after round $tn$, each agent $i \in N$ values all the agents' bundles equally.
Subsequently, any agent who has received a good in some round $r \in \{tn+1, \dots, tn+n\}$ must be envied by another agent who has not received a good since the beginning of round $tn+1$.
Therefore, the ECE algorithm allocates the goods $g_{tn+1}, \dots, g_{tn+n}$ to $n$ distinct agents, completing the induction.
This implies that at the end of $m$ rounds, each agent values all the agents' bundles equally, receiving a utility of exactly $1/n$.
It follows that the utilitarian welfare of this allocation is $1$.

On the other hand, suppose that the first $x$ goods are allocated to agent $1$, and each agent $i \in N \setminus \{1\}$ receives the goods $g_{x^{i-1} + 1}, \dots, g_{x^i}$. 
Then, agent $1$ receives a utility of $1-\epsilon$, each agent $i \in N \setminus \{1, n\}$ receives a utility of $(1-\epsilon)(1 - 1/x) = 1 - \epsilon - (1-\epsilon)/x$,
and agent $n$ receives a utility of $1-1/x \geq 1 - \epsilon - (1-\epsilon)/x$.
Hence, the maximum utilitarian welfare is at least $n(1-\epsilon) - (n-1)(1-\epsilon)/x$.

Therefore, the strong utilitarian price of this instance is at least  $n(1-\epsilon) - (n-1)(1-\epsilon)/x$, which approaches $n$ as $\epsilon \rightarrow 0$ and $x \rightarrow \infty$.

\medskip

We next derive the lower and upper bounds for the strong egalitarian price when $n = 2$ (\cref{thm:agent_good_with_max_util}(b)).
For the lower bound, consider an instance with $m \geq 3$ goods and the following utilities:
\begin{itemize}
    \item $u_1(g_j) = 1/2$ if $j \leq 2$ and $u_1(g_j) = 0$ otherwise.
    \item $u_2(g_j) = 1/m$ for all $j$.
\end{itemize}

Consider the allocation obtained by the ECE algorithm with the heuristic that chooses a pair of unenvied agent and unallocated good maximizing the utility in every round.
We assume that ties are broken by choosing the agent and good with the lowest indices.
Then, in the first round, $g_1$ is allocated to agent $1$.
This causes agent $2$ to envy agent $1$, so in the second round, $g_2$ is allocated to agent $2$.
At this point, agent $1$ values each of the two bundles exactly $1/2$.
Note that while new goods may be added to the bundles and the bundles may be swapped later, the ECE algorithm cannot take a good from one bundle and move it to the other bundle.
Since goods $g_1$ and $g_2$ are placed in different bundles in round~$2$, in the final allocation, agent $1$'s bundle can contain only one of $g_1$ and $g_2$. 
Hence, agent $1$'s utility, and therefore the egalitarian welfare, is at most $1/2$.

Now, consider instead the allocation where agent $1$ gets both goods $g_1$ and $g_2$, while the remaining goods are allocated to agent $2$.
Then, agent $1$ gets a utility of $1$ and agent~$2$ gets a utility of $1-2/m$.
The egalitarian welfare of this allocation is $1-2/m$.
Therefore, the strong egalitarian price of this instance is at least $(1-2/m)/(1/2)$, which approaches $2$ as $m\rightarrow \infty$.

We next prove that the strong egalitarian price when $n = 2$ is at most $3$.
If $m = 1$, every allocation has an egalitarian welfare of $0$, so the strong egalitarian price is $1$.
Take any instance with $n = 2$ agents and $m \geq 2$ goods, and let $(A_1, A_2)$ be an allocation obtained with the ECE algorithm using the heuristic that chooses a pair of unenvied agent and unallocated good maximizing the utility in every round.
Consider an agent with the minimum utility between the two agents; without loss of generality, assume that this is agent $2$.
Then, the egalitarian welfare of the allocation~$(A_1, A_2)$ is the utility of agent~$2$.

If agent $2$ does not envy agent $1$, then $1 \leq u_2(A_1) + u_2(A_2) \leq 2u_2(A_2)$, so $u_2(A_2) \geq 1/2$.
Since the maximum egalitarian welfare is at most $1$, the egalitarian welfare in this case is at least $1/2$ times the maximum.
From now on, assume that agent $2$ envies agent $1$.

Let~$g$ be the last good added to the bundle $A_1$.
Since agent $2$ envies agent $1$, good~$g$ must be allocated to agent $1$ and there are no bundle exchanges afterward.
On the other hand, agent $2$ must not envy agent $1$ immediately before good~$g$ is allocated.
Since agent $2$'s utility never decreases, $u_2(A_2) \geq u_2(A_1 \setminus \{g\})$.

We consider two cases.
First, suppose that $u_1(g) \geq u_2(g)$.
Observe that in any allocation, the utility of the agent who is not allocated good $g$ is at most\footnote{This does not hold under subadditive utilities.
In fact, \cref{thm:agent_good_with_max_util}(b) does not hold for such utilities---see Appendix~\ref{subsec:subadditive} for details.} $1-u_2(g)$.
Hence, the maximum egalitarian welfare is at most $1-u_2(g)$.
This means that $1-u_2(g) \leq u_2(M \setminus \{g\}) \leq u_2(A_1 \setminus \{g\}) + u_2(A_2) \leq 2u_2(A_2)$.
Therefore, agent $2$'s utility is at least $1/2$ times the maximum egalitarian welfare.

On the other hand, suppose that $u_1(g) < u_2(g)$.
Since good $g$ is allocated to agent $1$, according to the ECE heuristic used, agent $2$ must be envied by agent $1$ immediately before good~$g$ is allocated.
In particular, agent $2$ must own a non-empty bundle at this point.
Note that empty bundles are not involved in bundle exchanges.
Hence, the heuristic must have allocated some good $g'$ to agent $2$ prior to the round when good $g$ is allocated to agent $1$;
by definition of the heuristic, $u_2(g') \geq u_2(g)$.
Furthermore, since the ECE algorithm never decreases agent $2$'s utility, we have $u_2(A_2) \geq u_2(g') \geq u_2(g)$.
Therefore, $1 \leq u_2(g) + u_2(A_1 \setminus \{g\}) + u_2(A_2) \leq 3u_2(A_2)$.
Equivalently, $u_2(A_2) \geq 1/3$.
Since the maximum egalitarian welfare is at most $1$, it follows that agent $2$'s utility is at least $1/3$ times the maximum egalitarian welfare.
Hence, we have shown that in each case, the egalitarian welfare of the allocation---which is equal to agent $2$'s utility---is at least $1/3$ times the maximum egalitarian welfare.

\medskip

Lastly, we determine the strong egalitarian price when $n \geq 3$ (\cref{thm:agent_good_with_max_util}(c)).
Consider an instance with $n \geq 3$ agents, $m = n$ goods, and the following utilities:
\begin{center}
\begin{tabular}{c|cccccc}
    $g$      & $g_1$ & $\dots$ & $g_{n-3}$ & $g_{n-2}$ & $g_{n-1}$ & $g_n$ \\
    \hline
    $u_1(g)$ & $1$ & $\dots$ & $0$ & $0$ & $0$ & $0$ \\
    $\vdots$ & $\vdots$ & $\ddots$ & $\vdots$ & $\vdots$ & $\vdots$ & $\vdots$ \\
    $u_{n-3}(g)$ & $0$ & $\dots$ & $1$ & $0$ & $0$ & $0$ \\
    $u_{n-2}(g)$ & $0$ & $\dots$ & $0$ & $2/3$ & $1/3$ & $0$ \\
    $u_{n-1}(g)$ & $0$ & $\dots$ & $0$ & $0$ & $2/3$ & $1/3$ \\
    $u_n(g)$ & $0$ & $\dots$ & $0$ & $2/3$ & $1/3$ & $0$ \\
\end{tabular}
\end{center}

Consider the allocation produced by the ECE algorithm using the heuristic that chooses a pair of agent and good maximizing the utility in every round.
Assume that ties are broken by choosing the agent and good with the lowest indices.
For each $i \in \{1, \dots, n-3\}$, the heuristic must choose to allocate good~$g_i$ to agent $i$.
Then, in the last three rounds, goods $g_{n-2}, g_{n-1},$ and $g_n$ are allocated to agents $n-2, n-1,$ and $n$ respectively.
Note that only agent $n$ envies another agent, so no envy cycle exists.
Since agent $n$ receives a utility of $0$, the egalitarian welfare of this allocation is also $0$.

On the other hand, consider the allocation where goods $g_1, \dots, g_{n-3}$ are allocated to agents $1, \dots, n-3$ respectively, and goods $g_{n-2}, g_{n-1},$ and $g_n$ are allocated to agents $n-2, n,$ and $n-1$ respectively.
The egalitarian welfare of this allocation is $1/3$.
Hence, the strong egalitarian price of this instance is $(1/3)/0 = \infty$.
\end{proof}

\section{Max-min matching heuristic}
\label{sec:max_min_matching}

Thus far, all of our heuristics have a strong egalitarian price of $\infty$, which is the same as the ECE algorithm without any heuristic, when $n \geq 3$.
To attain a finite strong egalitarian price in this case, we next consider a non-greedy heuristic.

We adapt an algorithm by \citet{BezakovaDa05} for finding an allocation with egalitarian welfare at least $1/(m-n+1)$ times the optimum.\footnote{While there are algorithms for maximizing the egalitarian welfare with better approximation factors (e.g., \citep{Golovin05,AsadpourSa10}), they achieve the improved factors either through randomization or by excluding some subset of agents.
Hence, these algorithms are not useful for our purposes.
}
The main idea of the algorithm involves \emph{max-min matching}, a fundamental problem in graph theory.
A \emph{matching} of a graph is a subset of edges with no common vertices, and a max-min matching of an edge-weighted bipartite graph is a matching that has  maximum cardinality, and subject to that, the minimum edge weight is maximized.
An instance is first modeled as a complete bipartite graph with the vertex set $N \cup M$, and there is an edge $\{i, g\}$ of weight $u_i(g)$ between every pair $i \in N$ and $g \in M$.
Then, a max-min matching of this graph corresponds to a partial allocation of the original instance where every agent receives exactly one good, and the egalitarian welfare is maximum among such partial allocations.

We define \emph{the max-min matching heuristic} as follows.
If $m < n$, first add $n-m$ dummy goods valued at $0$ by every agent.
The heuristic computes a sequence of $n$ distinct goods $g'_1, \dots, g'_n \in M$ such that $\min_{i \in N} u_i(g'_i)$ is maximized---this can be done in polynomial time~\citep[Section 2]{BezakovaDa05}.
Then, in each round $i \in \{1, \dots, n\}$ of the ECE algorithm, good $g'_i$ is allocated to agent $i$ (who must be unenvied), so that the egalitarian welfare of the partial allocation obtained at the end of round~$n$ is maximized.
The choices in the remaining rounds can be arbitrary.

\begin{restatable}{theorem}{thmmaxminmatching}
\label{thm:max_min_matching}
Consider the ECE algorithm using the max-min matching heuristic.
\begin{enumerate}[label=(\alph*)]
    \item Its strong utilitarian price is $n^2$.
    \item Its strong egalitarian price is $m-n+1$ if $m \geq n$, and $1$ otherwise.
\end{enumerate}
\end{restatable}

In contrast, for all other heuristics studied in this paper, the infinite strong egalitarian price holds even when $n$ and $m$ are fixed.

To prove the upper bound in \cref{thm:max_min_matching}(b), we will use the following lemma, which was shown by \citet{BezakovaDa05} for additive utilities.\footnote{This is Lemma 2 in the extended version of their work.}
We include here a proof that also applies for subadditive utilities.

\begin{restatable}[\citealp{BezakovaDa05}]{lemma}{lemmaxminmatching}
\label{lem:max_min_matching}
Given any instance with $n \geq 2$ agents and $m \geq n$ goods, let $g'_1, \dots, g'_n$ be distinct goods such that $\min_{i \in N} u_i(g'_i)$ is maximized.
Then, the maximum egalitarian welfare of the instance is at most $(m-n+1) \cdot \min_{i \in N} u_i(g'_i)$.
\end{restatable}

\begin{proof}
Let $\mathcal{A} = (A_1, \dots, A_n)$ be an allocation with maximum egalitarian welfare.
If there is some agent $i \in N$ such that $A_i = \emptyset$, then the maximum egalitarian welfare is $0$, which is at most $(m-n+1) \cdot \min_{i \in N} u_i(g'_i)$.
We can now assume that $A_i \neq \emptyset$ for any $i \in N$.
For each $i \in N$, let $g''_i$ be some good in $A_i$ such that $u_i(g''_i)$ is maximized.
By the choice of $g'_1, \dots, g'_n$, we have $\min_{i \in N} u_i(g'_i) \geq \min_{i \in N} u_i(g''_i)$.
Now fix agent $k$ such that $u_k(g''_k) = \min_{i \in N} u_i(g''_i)$.
Observe that
\begin{align*}
 u_k(A_k) \leq \sum_{g \in A_k} u_k(g) \leq |A_k|\cdot u_k(g''_k) \leq (m-n+1) \cdot u_k(g''_k),
\end{align*}
where the last inequality follows from the fact that $|A_i| \geq 1$ for any $i \in N \setminus \{k\}$, and so $|A_k| \leq m-n+1$.
Since $u_k(A_k)$ is at least the egalitarian welfare of allocation $\mathcal{A}$, and $u_k(g''_k) = \min_{i \in N} u_i(g''_i)$ is at most $\min_{i \in N} u_i(g'_i)$, it follows that the maximum egalitarian welfare of the given instance is at most $(m-n+1) \cdot \min_{i \in N} u_i(g'_i)$.
\end{proof}

We are now ready to prove \cref{thm:max_min_matching}.

\begin{proof}[Proof of \cref{thm:max_min_matching}]
We start by finding the strong egalitarian price.
If $m < n$, the strong egalitarian price is $1$ since every allocation has egalitarian welfare $0$.
For the case where $m \geq n$, the upper bound follows directly from \cref{lem:max_min_matching} since the ECE algorithm never decreases the utility of any agent, and thus also the egalitarian welfare.

We next prove the lower bound.
If $m = n$, then $m-n+1 = 1$ and the lower bound holds trivially.
Assume henceforth that $m > n$.
Let $k = m-n+1 \geq 2$ and $0 < \epsilon < 1/4$.
Consider the following instance with $n$~agents and $m$~goods:
\begin{itemize}
    \item For each $i \in N \setminus \{n-1, n\}$, $u_i(g_i) = 1$.
    \item $u_{n-1}(g_{n-1}) = 1-1/(2k)-\epsilon$ and $u_{n-1}(g_n) = 1/(2k)+\epsilon$.
    \item $u_n(g_{n-1}) = 1/2$, and $u_n(g_j) = 1/(2k)$ for $j \geq n$.
    \item $u_i(g_j) = 0$ for all other pairs $(i, j)$.
\end{itemize}

The expression~$\min_{i \in N} u_i(g'_i)$ is maximized if and only if $g'_i = g_i$ for each $i \in N \setminus \{n-1, n\}$, $g'_{n-1} = g_n$, and $g'_n = g_{n-1}$.
In this case, $\min_{i \in N} u_i(g'_i) = u_{n-1}(g_n) = 1/(2k)+\epsilon$.
Hence, the max-min matching heuristic allocates $g'_i$ to agent~$i$ in each round~$i\in\{1,\dots,n\}$.
In particular, agent~$n$ receives good $g_{n-1}$ of utility~$1/2$.
Since her utility never decreases throughout the execution of the ECE algorithm, agent~$n$ will never envy another agent after round~$n$.
This means that agent~$n$'s bundle will never be passed to a different agent, and good $g_{n-1}$ will stay in agent~$n$'s possession until the end of the algorithm.
On the other hand, agent $(n-1)$'s utility at the end of the algorithm is at most $u_{n-1}(M \setminus \{g_{n-1}\}) = 1/(2k) + \epsilon$.
It follows that any allocation of this instance produced by the ECE algorithm with the max-min matching heuristic has an egalitarian welfare of at most $1/(2k)+\epsilon$.

Now, consider the allocation where each agent~$i \in N \setminus \{n\}$ receives good~$g_i$, while agent~$n$ receives goods~$g_n, \dots, g_m$.
In this allocation, each agent~$i \in N \setminus \{n-1, n\}$ receives a utility of~$1$, agent $n-1$ receives a utility of $1-1/(2k)-\epsilon$, and agent~$n$ receives a total utility of $k \cdot 1/(2k) = 1/2$.
Since $k \geq 2$ and $\epsilon < 1/4$, we have $1-1/(2k)-\epsilon > 1/2$.
Hence, the egalitarian welfare of this allocation is $1/2$.

Therefore, the strong egalitarian price of this instance is at least $(1/2)/(1/(2k)+\epsilon)$, which approaches $k = m-n+1$ as $\epsilon$ approaches $0$.
By the upper bound established earlier, it follows that the strong egalitarian price of the ECE algorithm with the max-min matching heuristic is exactly $m-n+1$.

We continue by finding the strong utilitarian price.
The upper bound follows from the strong utilitarian price of the vanilla ECE algorithm (\cref{prop:strong_price}), since the set of allocations that can be obtained with any ECE heuristic is a subset of all allocations that can be produced by the ECE algorithm.

For the lower bound, consider an instance with $n \geq 2$ agents, $m = kn-1$ goods where $k \geq 2$ is an integer, and the following utilities:
\begin{itemize}
    \item $u_1(g_j) = 1/m$ for all $j$.
    \item For each $i \in N \setminus \{1\}$, $u_i(g_i) = 1/m$ and $u_i(g_{(k-1)n+i-1}) = 1-1/m$.
    \item $u_i(g_j) = 0$ for all other pairs $(i, j)$.
\end{itemize}

Consider the following allocation obtained by the max-min matching heuristic:
In each round $i \in \{1, \dots, n\}$, we allocate good~$g'_i = g_i$ to agent $i$.
Then, $\min_{i \in N} u_i(g'_i) = 1/m$, which is the maximum egalitarian welfare for any partial allocation where each agent receives exactly one good, since $u_1(g) = 1/m$ for any good $g$.
In the subsequent rounds, the goods are allocated in increasing order of their indices from $g_{n+1}$ to $g_m$ such that for each $t \in \{1, \dots, k-1\}$ and $i \in \{1, \dots, n\}$, good $g_{tn+i}$ is given to agent $i$ (except when $t = k-1$ and $i = n$, since $g_{kn}$ does not exist).

In the first $(k-1)n$ rounds, the intermediate allocations remain envy-free.
After round $(k-1)n+i$ for each $i \in \{1, \dots, n-1\}$, agent $i+1$ begins to envy agent $i$, and continues to envy her until the end of the algorithm.
Hence, the envy graph remains acyclic throughout.
At the end, agent $1$ gets a total utility of $k/m$ and each of the remaining agents gets a total utility of $1/m$.
Therefore, the utilitarian welfare of this allocation is $(k+n-1)/m$.

On the other hand, consider the allocation where each agent $i \in N \setminus \{1\}$ only gets the good $g_{(k-1)n+i-1}$, and agent~$1$ gets the remaining goods $g_1, \dots, g_{(k-1)n}$.
In this case, agent~$1$ gets a total utility of $(k-1)n/m = 1-(n-1)/m$, while each of the remaining agents gets a total utility of $1-1/m$. Hence, the utilitarian welfare of this allocation is $n-(2n-2)/m$.

Therefore, the strong utilitarian price of this instance is at least 
\begin{align*}
    \frac{n-\frac{2n-2}{m}}{\frac{k+n-1}{m}} &= \frac{nm-2n+2}{k+n-1} \\
    &= \frac{n(kn-1)-2n+2}{k+n-1} \tag{since $m = kn-1$} \\
    &= \frac{kn^2-3n+2}{k+n-1} \\
    &= n^2 - \frac{n^3-n^2+3n-2}{k+n-1},
\end{align*}
which approaches $n^2$ as $k\rightarrow\infty$.
\end{proof}

\section{Experiments}
\label{sec:experiments}

\begin{figure*}[t!]
\centering
\fontsize{6}{6}\selectfont
\tikzset{every mark/.append style={scale=0.5}} 
\begin{tikzpicture}[scale=1.12]
    \begin{axis}[
        name = random2util,
        height = 4cm, width = 4cm,
        align = center,
        title = \fontsize{7}{6}\selectfont{Utilitarian price with $2$ agents},
        xlabel = {Number of goods},
        ylabel = {Utilitarian price},
        xmin = 2, xmax = 9, ymin = 1, ymax = 1.2,
        xtick = {2, 3, 4, 5, 6, 7, 8, 9},
        ytick = {1, 1.1, 1.2},
        minor y tick num = 4,
        legend style = {
            at = {(0cm, -0.8cm)},
            anchor = north west,
            cells = {anchor = west},
            legend columns = 3,
            legend to name = {leg:random},
            /tikz/every even column/.append style={column sep=0.2cm},
        }
    ]
    
    \addplot[color = black, mark = *]
        coordinates {
            (2, 1.07)
            (3, 1.11)
            (4, 1.14)
            (5, 1.15)
            (6, 1.17)
            (7, 1.18)
            (8, 1.19)
            (9, 1.20)
        };
        \addlegendentry{No heuristic;}

    \addplot[color = red, mark = x]
        coordinates {
            (2, 1.00)
            (3, 1.05)
            (4, 1.07)
            (5, 1.07)
            (6, 1.07)
            (7, 1.07)
            (8, 1.07)
            (9, 1.06)
        };
        \addlegendentry{Agent with maximum utility;}
        
    \addplot[color = YellowOrange, mark = +]
        coordinates {
            (2, 1.07)
            (3, 1.07)
            (4, 1.08)
            (5, 1.08)
            (6, 1.07)
            (7, 1.07)
            (8, 1.06)
            (9, 1.06)
        };
        \addlegendentry{Good with maximum utility;}

    \addplot[color = Green, mark = o]
        coordinates {
            (2, 1.00)
            (3, 1.07)
            (4, 1.11)
            (5, 1.14)
            (6, 1.17)
            (7, 1.18)
            (8, 1.19)
            (9, 1.20)
        };
        \addlegendentry{Agent \& good with maximum utility (first round only);}
    
    \addplot[color = blue, mark = square*]
        coordinates {
            (2, 1.00)
            (3, 1.02)
            (4, 1.02)
            (5, 1.02)
            (6, 1.01)
            (7, 1.01)
            (8, 1.01)
            (9, 1.01)
        };
        \addlegendentry{Agent \& good with maximum utility;}
    
    \addplot[color = Orchid, mark = triangle]
        coordinates {
            (2, 1.00)
            (3, 1.06)
            (4, 1.11)
            (5, 1.13)
            (6, 1.15)
            (7, 1.16)
            (8, 1.17)
            (9, 1.18)
        };
        \addlegendentry{Max-min matching}
    \end{axis}

    \begin{axis}[
        name = random2egal,
        at = {($(random2util.east)+(1.05cm,0)$)},
        anchor = west,
        height = 4cm, width = 4cm,
        align = center,
        title = \fontsize{7}{6}\selectfont{Egalitarian price with $2$ agents},
        xlabel = {Number of goods},
        ylabel = {Egalitarian price},
        xmin = 2, xmax = 9, ymin = 1, ymax = 1.3,
        xtick = {2, 3, 4, 5, 6, 7, 8, 9},
        minor y tick num = 4,
    ]
    
    \addplot[color = black, mark = *]
        coordinates {
            (2, 1.12)
            (3, 1.19)
            (4, 1.23)
            (5, 1.24)
            (6, 1.26)
            (7, 1.27)
            (8, 1.28)
            (9, 1.29)
        };

    \addplot[color = red, mark = x]
        coordinates {
            (2, 1.00)
            (3, 1.11)
            (4, 1.12)
            (5, 1.13)
            (6, 1.13)
            (7, 1.13)
            (8, 1.13)
            (9, 1.13)
        };
        
    \addplot[color = YellowOrange, mark = +]
        coordinates {
            (2, 1.12)
            (3, 1.08)
            (4, 1.15)
            (5, 1.19)
            (6, 1.22)
            (7, 1.24)
            (8, 1.25)
            (9, 1.26)
        };

    \addplot[color = Green, mark = o]
        coordinates {
            (2, 1.00)
            (3, 1.07)
            (4, 1.19)
            (5, 1.23)
            (6, 1.26)
            (7, 1.27)
            (8, 1.28)
            (9, 1.29)
        };
    
    \addplot[color = blue, mark = square*]
        coordinates {
            (2, 1.00)
            (3, 1.03)
            (4, 1.05)
            (5, 1.06)
            (6, 1.07)
            (7, 1.08)
            (8, 1.08)
            (9, 1.08)
        };
    
    \addplot[color = Orchid, mark = triangle*]
        coordinates {
            (2, 1.00)
            (3, 1.08)
            (4, 1.18)
            (5, 1.22)
            (6, 1.25)
            (7, 1.26)
            (8, 1.27)
            (9, 1.28)
        };
    \end{axis}

    \begin{axis}[
        name = random3util,
        at = {($(random2egal.east)+(1.05cm,0)$)},
        anchor = west,
        height = 4cm, width = 4cm,
        align = center,
        title = \fontsize{7}{6}\selectfont{Utilitarian price with $3$ agents},
        xlabel = {Number of goods},
        ylabel = {Utilitarian price},
        xmin = 3, xmax = 9, ymin = 1, ymax = 1.24,
        xtick = {3, 4, 5, 6, 7, 8, 9},
        minor y tick num = 4,
    ]
    
    \addplot[color = black, mark = *]
        coordinates {
            (3, 1.11)
            (4, 1.14)
            (5, 1.17)
            (6, 1.20)
            (7, 1.21)
            (8, 1.23)
            (9, 1.24)
        };

    \addplot[color = red, mark = x]
        coordinates {
            (3, 1.03)
            (4, 1.08)
            (5, 1.09)
            (6, 1.10)
            (7, 1.10)
            (8, 1.11)
            (9, 1.11)
        };
        
    \addplot[color = YellowOrange, mark = +]
        coordinates {
            (3, 1.11)
            (4, 1.10)
            (5, 1.10)
            (6, 1.10)
            (7, 1.10)
            (8, 1.09)
            (9, 1.08)
        };

    \addplot[color = Green, mark = o]
        coordinates {
            (3, 1.05)
            (4, 1.11)
            (5, 1.15)
            (6, 1.18)
            (7, 1.20)
            (8, 1.22)
            (9, 1.24)
        };
    
    \addplot[color = blue, mark = square*]
        coordinates {
            (3, 1.03)
            (4, 1.04)
            (5, 1.04)
            (6, 1.04)
            (7, 1.03)
            (8, 1.02)
            (9, 1.02)
        };
    
    \addplot[color = Orchid, mark = triangle*]
        coordinates {
            (3, 1.04)
            (4, 1.09)
            (5, 1.14)
            (6, 1.17)
            (7, 1.19)
            (8, 1.21)
            (9, 1.22)
        };
    \end{axis}

    \begin{axis}[
        name = random3egal,
        at = {($(random3util.east)+(1.05cm,0)$)},
        anchor = west,
        height = 4cm, width = 4cm,
        align = center,
        title = \fontsize{7}{6}\selectfont{Egalitarian price with $3$ agents},
        xlabel = {Number of goods},
        ylabel = {Egalitarian price},
        xmin = 3, xmax = 9, ymin = 1, ymax = 1.4,
        xtick = {3, 4, 5, 6, 7, 8, 9},
        minor y tick num = 4,
    ]
    
    \addplot[color = black, mark = *]
        coordinates {
            (3, 1.21)
            (4, 1.27)
            (5, 1.29)
            (6, 1.36)
            (7, 1.37)
            (8, 1.38)
            (9, 1.39)
        };

    \addplot[color = red, mark = x]
        coordinates {
            (3, 1.06)
            (4, 1.18)
            (5, 1.21)
            (6, 1.24)
            (7, 1.24)
            (8, 1.24)
            (9, 1.25)
        };
        
    \addplot[color = YellowOrange, mark = +]
        coordinates {
            (3, 1.24)
            (4, 1.13)
            (5, 1.18)
            (6, 1.27)
            (7, 1.28)
            (8, 1.29)
            (9, 1.32)
        };

    \addplot[color = Green, mark = o]
        coordinates {
            (3, 1.15)
            (4, 1.19)
            (5, 1.21)
            (6, 1.32)
            (7, 1.35)
            (8, 1.36)
            (9, 1.38)
        };
    
    \addplot[color = blue, mark = square*]
        coordinates {
            (3, 1.13)
            (4, 1.09)
            (5, 1.11)
            (6, 1.14)
            (7, 1.14)
            (8, 1.15)
            (9, 1.16)
        };
    
    \addplot[color = Orchid, mark = triangle*]
        coordinates {
            (3, 1.00)
            (4, 1.10)
            (5, 1.20)
            (6, 1.32)
            (7, 1.34)
            (8, 1.35)
            (9, 1.38)
        };
    \end{axis}
\end{tikzpicture}
\ref{leg:random}
\caption{Average utilitarian and egalitarian prices under the uniform distribution on $[0,1]$}
\label{fig:uniform}
\end{figure*}
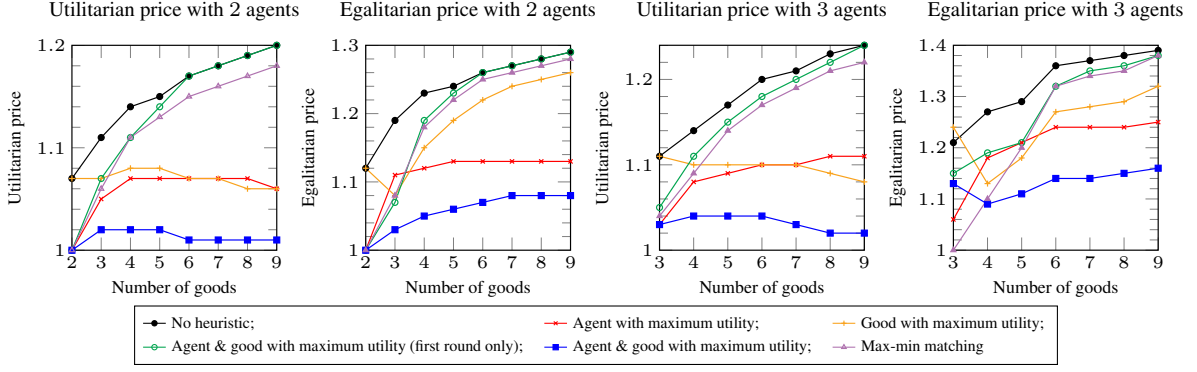

In this section, we complement our theoretical worst-case analysis with experimental average-case results.
For each $n \in \{2, 3\}$ and $m \in \{n, n+1, \dots, 9\}$, we generated $10{,}000$ instances by drawing each agent $i$'s utility $u_i(g)$ for each good $g$ independently from a uniform distribution.\footnote{We also ran the experiments using exponential and log-normal distributions, as well as real-world data from Spliddit~\citep{GoldmanPr14,Shah17}; see \Cref{fig:exponential,fig:lognormal,fig:spliddit} in Appendix~\ref{ap:experiments}.\label{foot:experiments}}
We then normalized the utilities by dividing $u_i(g)$ with $u_i(M)$, so that each agent's utilities for all the goods sum up to $1$.
For each instance and each of the six heuristics, we computed the ECE allocation deterministically---breaking any ties by choosing the lowest-index unenvied agent and, subject to that, the lowest-index available good---and calculated the utilitarian and egalitarian prices based on this allocation.
We then obtained the ``average'' utilitarian (resp.~egalitarian) price\footnote{Since we consider a specific allocation, we omit ``strong'' from the price terminology.} of each heuristic over the $10{,}000$ instances, by finding the \emph{harmonic mean}, which is defined as the reciprocal of the arithmetic mean of the reciprocals.
The harmonic mean is more suitable than the arithmetic mean for our purposes, since our numbers are \emph{ratios}---the values of which can be arbitrarily high for the egalitarian price---between the optimal welfare and the welfare obtained using each heuristic.
While the arithmetic mean can be heavily skewed by a handful of instances with particularly high egalitarian prices, the effect of such instances on the harmonic mean is limited, as the  inverse of the price is bounded in $[0, 1]$.
The obtained values are presented in \Cref{fig:uniform}. 

As one may expect, the heuristic that chooses an unenvied agent~$i$ and an unallocated good~$g$ with the highest utility~$u_i(g)$ in every round has the lowest average utilitarian price for each $n \in \{2, 3\}$ and $m \in \{n, \dots, 9\}$, as well as the lowest average egalitarian price when $n = 2$ and $m \in \{2, \dots, 9\}$.
More interestingly, this heuristic also has the lowest average egalitarian price when $n = 3$ and $m \in \{4, \dots, 9\}$, despite its strong egalitarian price of $\infty$ in the case of three agents.
In general, the heuristics that optimize the welfare in every round (even if only for the agents or only for the goods) fare better in terms of the average welfare than those optimizing only in the early rounds (such as the max-min matching heuristic), in spite of the fact that the former heuristics have similar prices to the vanilla ECE algorithm.
Further discussion of the experimental results is deferred to Appendix~\ref{ap:experiments}.

\section{Conclusion}
\label{sec:conclusion}

In this paper, we have explored various heuristics for the envy cycle elimination (ECE) algorithm, widely regarded as one of the most important algorithms in fair division.
We demonstrated that it is possible to improve the welfare guarantees of the algorithm via carefully chosen heuristics.
In particular, the greedy heuristic that chooses a pair of good and agent maximizing the utility in each round  outperforms the vanilla algorithm in terms of the strong utilitarian price and, for two agents, the strong egalitarian price.
Moreover, our experiments indicate that this heuristic also has the best average-case performance, and that our other proposed heuristics significantly outperform the vanilla algorithm despite some having comparable strong prices.

A natural follow-up would be to characterize the ECE heuristics that provide the optimal strong utilitarian and egalitarian prices, subject to polynomial-time computation.
In particular, when there are at least three agents, it remains open whether there exist heuristics that improve upon the vanilla ECE algorithm in terms of both prices simultaneously.
Another possible direction is to investigate ECE heuristics determining which envy cycles to resolve. An example of such heuristics is the top-trading ECE algorithm used to obtain EF1 \emph{chore} allocation~\citep{BhaskarSrVa21}.

\bibliographystyle{plainnat}
\bibliography{arxiv}

\begin{thebibliography}{21}
\providecommand{\natexlab}[1]{#1}
\providecommand{\url}[1]{\texttt{#1}}
\expandafter\ifx\csname urlstyle\endcsname\relax
  \providecommand{\doi}[1]{doi: #1}\else
  \providecommand{\doi}{doi: \begingroup \urlstyle{rm}\Url}\fi

\bibitem[Amanatidis et~al.(2023)Amanatidis, Aziz, Birmpas, Filos-Ratsikas, Li, Moulin, Voudouris, and Wu]{AmanatidisAzBi23}
Georgios Amanatidis, Haris Aziz, Georgios Birmpas, Aris Filos-Ratsikas, Bo~Li, Hervé Moulin, Alexandros~A. Voudouris, and Xiaowei Wu.
\newblock Fair division of indivisible goods: Recent progress and open questions.
\newblock \emph{Artificial Intelligence}, 322:\penalty0 103965, 2023.

\bibitem[Amanatidis et~al.(2026)Amanatidis, Birmpas, and Reiffenh\"{a}user]{AmanatidisBiRe26}
Georgios Amanatidis, Georgios Birmpas, and Rebecca Reiffenh\"{a}user.
\newblock Envy cycle elimination with strategic agents: Best responses and fairness guarantees.
\newblock \emph{CoRR}, abs/2605.31253, 2026.

\bibitem[Asadpour and Saberi(2010)]{AsadpourSa10}
Arash Asadpour and Amin Saberi.
\newblock An approximation for max-min fair allocation of indivisible goods.
\newblock \emph{SIAM Journal on Computing}, 39\penalty0 (7):\penalty0 2970--2989, 2010.

\bibitem[Barman and Krishnamurthy(2020)]{BarmanKr20}
Siddharth Barman and Sanath~Kumar Krishnamurthy.
\newblock Approximation algorithms for maximin fair division.
\newblock \emph{ACM Transactions on Economics and Computation}, 8\penalty0 (1):\penalty0 5:1--5:28, 2020.

\bibitem[Bei et~al.(2021)Bei, Lu, Manurangsi, and Suksompong]{BeiLuMa21}
Xiaohui Bei, Xinhang Lu, Pasin Manurangsi, and Warut Suksompong.
\newblock The price of fairness for indivisible goods.
\newblock \emph{Theory of Computing Systems}, 65\penalty0 (7):\penalty0 1069--1093, 2021.

\bibitem[Bertsimas et~al.(2011)Bertsimas, Farias, and Trichakis]{BertsimasFaTr11}
Dimitris Bertsimas, Vivek Farias, and Nikolaos Trichakis.
\newblock The price of fairness.
\newblock \emph{Operations Research}, 59\penalty0 (1):\penalty0 17--31, 2011.

\bibitem[Bez\'{a}kov\'{a} and Dani(2005)]{BezakovaDa05}
Ivona Bez\'{a}kov\'{a} and Varsha Dani.
\newblock Allocating indivisible goods.
\newblock \emph{ACM SIGecom Exchanges}, 5\penalty0 (3):\penalty0 11--18, 2005.
\newblock Extended version available at \url{https://newtraell.cs.uchicago.edu/files/tr_authentic/TR-2004-10.pdf}.

\bibitem[Bhaskar et~al.(2021)Bhaskar, Sricharan, and Vaish]{BhaskarSrVa21}
Umang Bhaskar, A.~R. Sricharan, and Rohit Vaish.
\newblock On approximate envy-freeness for indivisible chores and mixed resources.
\newblock In \emph{Proceedings of the 24th International Conference on Approximation Algorithms for Combinatorial Optimization Problems (APPROX)}, pages 1:1--1:23, 2021.

\bibitem[Brams et~al.(2003)Brams, Edelman, and Fishburn]{BramsEdFi03}
Steven~J. Brams, Paul~H. Edelman, and Peter~C. Fishburn.
\newblock Fair division of indivisible items.
\newblock \emph{Theory and Decision}, 55\penalty0 (2):\penalty0 147--180, 2003.

\bibitem[Caragiannis et~al.(2012)Caragiannis, Kaklamanis, Kanellopoulos, and Kyropoulou]{CaragiannisKaKa12}
Ioannis Caragiannis, Christos Kaklamanis, Panagiotis Kanellopoulos, and Maria Kyropoulou.
\newblock The efficiency of fair division.
\newblock \emph{Theory of Computing Systems}, 50\penalty0 (4):\penalty0 589--610, 2012.

\bibitem[Celine et~al.(2023)Celine, Dzulfikar, and Koswara]{CelineDzKo23}
Karen~Frilya Celine, Muhammad~Ayaz Dzulfikar, and Ivan~Adrian Koswara.
\newblock Egalitarian price of fairness for indivisible goods.
\newblock In \emph{Proceedings of the 20th Pacific Rim International Conference on Artificial Intelligence (PRICAI)}, pages 23--28, 2023.

\bibitem[Demko and Hill(1988)]{DemkoHi88}
Stephen Demko and Theodore~P. Hill.
\newblock Equitable distribution of indivisible objects.
\newblock \emph{Mathematical Social Sciences}, 16\penalty0 (2):\penalty0 145--158, 1988.

\bibitem[Feldman et~al.(2024)Feldman, Mauras, Narayan, and Ponitka]{FeldmanMaNa24}
Michal Feldman, Simon Mauras, Vishnu~V. Narayan, and Tomasz Ponitka.
\newblock Breaking the envy cycle: Best-of-both-worlds guarantees for subadditive valuations.
\newblock In \emph{Proceedings of the 25th ACM Conference on Economics and Computation (EC)}, pages 1236--1266, 2024.

\bibitem[Goldman and Procaccia(2014)]{GoldmanPr14}
Jonathan Goldman and Ariel~D. Procaccia.
\newblock Spliddit: Unleashing fair division algorithms.
\newblock \emph{{ACM} {SIG}ecom Exchanges}, 13\penalty0 (2):\penalty0 41--46, 2014.

\bibitem[Golovin(2005)]{Golovin05}
Daniel Golovin.
\newblock Max-min fair allocation of indivisible goods.
\newblock SCS Technical Report Collection CMU-CS-05-144, Carnegie Mellon University, 2005.

\bibitem[Lipton et~al.(2004)Lipton, Markakis, Mossel, and Saberi]{LiptonMaMo04}
Richard~J. Lipton, Evangelos Markakis, Elchanan Mossel, and Amin Saberi.
\newblock On approximately fair allocations of indivisible goods.
\newblock In \emph{Proceedings of the 5th ACM Conference on Electronic Commerce (EC)}, pages 125--131, 2004.

\bibitem[Markakis and Santorinaios(2023)]{MarkakisSa23}
Evangelos Markakis and Christodoulos Santorinaios.
\newblock Improved {EFX} approximation guarantees under ordinal-based assumptions.
\newblock In \emph{Proceedings of the 22nd International Conference on Autonomous Agents and Multiagent Systems (AAMAS)}, pages 591--599, 2023.

\bibitem[Moulin(2003)]{Moulin03}
Herv\'{e} Moulin.
\newblock \emph{Fair Division and Collective Welfare}.
\newblock MIT Press, 2003.

\bibitem[Moulin(2019)]{Moulin19}
Herv\'{e} Moulin.
\newblock Fair division in the internet age.
\newblock \emph{Annual Review of Economics}, 11:\penalty0 407--441, 2019.

\bibitem[Plaut and Roughgarden(2020)]{PlautRoughgarden20}
Benjamin Plaut and Tim Roughgarden.
\newblock Almost envy-freeness with general valuations.
\newblock \emph{SIAM Journal on Discrete Mathematics}, 34\penalty0 (2):\penalty0 1039--1068, 2020.

\bibitem[Shah(2017)]{Shah17}
Nisarg Shah.
\newblock Spliddit: Two years of making the world fairer.
\newblock \emph{XRDS: Crossroads, The ACM Magazine for Students}, 24\penalty0 (1):\penalty0 24--28, 2017.

\end{thebibliography}

\appendix

\section{Beyond additive utilities}
\label{ap:nonadditive}

While we mainly focus on additive utilities in this paper, the ECE algorithm itself guarantees EF1 for the more general class of \emph{monotone} utilities~\citep{LiptonMaMo04}.
In this section, we shall discuss some results under utilities that are monotone but not necessarily additive.
In particular, we show that optimizing the welfare guarantee of the ECE algorithm under monotone utilities is infeasible, in the sense that for any $\epsilon > 0$, there exists an instance with monotone utilities where all ECE allocations have a utilitarian welfare of at most $\epsilon$ times the optimum.
Likewise, there exists a monotone instance with a positive maximum egalitarian welfare, where all ECE allocations have an egalitarian welfare of~$0$.
On the other hand, for the class of subadditive utilities, which lies between additive and monotone utilities, most of our results for additive utilities continue to hold.

We now define relevant concepts in the context of monotone utilities.
An instance has \emph{monotone utilities} if the utility function for each agent $i \in N$ satisfies $u_i(S) \leq u_i(T)$ for any subsets $S \subseteq T \subseteq M$.
In this section, we assume that the instance has monotone utilities, and for each $i \in N$, it holds that $u_i(\emptyset) = 0$ and $u_i(M) = 1$.
If the utility function for each agent $i \in N$ also satisfies $u_i(S \cup T) \leq u_i(S) + u_i(T)$ for any disjoint subsets $S, T \subseteq M$, we say that the utilities are \emph{subadditive}.
Note that the class of subadditive utilities includes the class of additive utilities, by replacing the inequality relation with equality.
Given a bundle $S \subseteq M$, the \emph{marginal utility} of good~$g$ for agent~$i$ with respect to $S$ is the difference $u_i(S \cup \{g\}) - u_i(S)$.
In particular, when $S = \emptyset$, the marginal utility of good~$g$ for agent~$i$ with respect to $S$ is exactly $u_i(\emptyset \cup \{g\}) - u_i(\emptyset) = u_i(\{g\})$.
We can then redefine our greedy heuristics to instead choose an unenvied agent~$i$, an unallocated good~$g$, or both, that maximize the marginal utility $u_i(A'_i \cup \{g\}) - u_i(A'_i)$, where $A'_i$ is the bundle currently owned by agent~$i$.
Furthermore, the max-min matching heuristic under non-additive utilities is defined to be the heuristic that allocates good $g'_i$ to agent $i$ in the $i$-th round for each $i \in N$, such that $\min_{i \in N} u_i(\{g'_i\})$ is maximized.\footnote{Note that each agent $i$ has an empty bundle prior to round $i$.
Hence, $u_i(\{g'_i\}) = u_i(\emptyset \cup \{g'_i\}) - u_i(\emptyset)$ is the marginal utility of good $g'_i$ to agent $i$ during round $i$.}

\subsection{Monotone utilities}
\label{subsec:monotone}

In the following two propositions, we show that the weak utilitarian and egalitarian prices of ECE are $\infty$ under monotone utilities.
This implies that for any ECE heuristic, the strong utilitarian and egalitarian prices under monotone utilities are also $\infty$.

\begin{proposition}
Under monotone utilities, the weak utilitarian price of EF1, and thus the weak utilitarian price of ECE, is $\infty$.
\end{proposition}

\begin{proof}
Consider any instance with $n \geq 2$ agents, $m \geq 2$ goods, and the following utilities:
Let $0 < \epsilon < 1$.
For each agent $i \in N$,
\begin{align*}
    u_i(S) = \begin{cases}
        1 &\text{if $S = M$,} \\
        \epsilon &\text{if $\emptyset \subsetneq S \subsetneq M$, and} \\
        0 &\text{if $S = \emptyset$.}
    \end{cases}
\end{align*}

Now suppose that some agent $i \in N$ has a utility of strictly more than $\epsilon$.
Then, she must be allocated the whole bundle of goods $M$.
This means that some agent $j \in N$ owns an empty bundle.
For any good $g \in M$, we have $u_j(A_i \setminus \{g\}) = \epsilon > 0 = u_j(A_j)$.
In other words, agent $j$ is not EF1 toward agent $i$.
By contraposition, it follows that in any EF1 allocation, every agent has a utility of at most~$\epsilon$.
Therefore, the maximum utilitarian welfare of any EF1 allocation is $n\epsilon$.
Since it is possible to obtain a utilitarian welfare of $1$ by allocating the entire bundle of goods $M$ to a single agent, the weak utilitarian price of EF1 for this instance is at least $1/(n\epsilon)$, which approaches $\infty$ as $\epsilon \to 0$.
Furthermore, since any allocation obtained by the ECE algorithm is necessarily EF1 under monotone utilities, the weak utilitarian price of ECE is also $\infty$.
\end{proof}

\begin{proposition}
Under monotone utilities, the weak egalitarian price of EF1, and thus the weak egalitarian price of ECE, is $\infty$.
\end{proposition}

\begin{proof}
Consider any instance with $n \geq 2$ agents, $m = n+2$ goods, and the following utilities:
\begin{itemize}
    \item $u_1(S) = |S \cap \{g_1, g_2, g_3, g_4\}|/4$.
    \item $u_2(S) = 1$ if $|S \cap \{g_1, g_2, g_3, g_4\}| \geq 3$, and $u_2(S) = 0$ otherwise.
    \item For each agent $i \in N \setminus \{1, 2\}$, $u_i(S) = 1$ if $g_{i+2} \in S$, and $u_i(S) = 0$ otherwise.
\end{itemize}

Let $\mathcal{A} = (A_1, \dots, A_n)$ be an allocation with positive egalitarian welfare.
In particular, $u_2(A_2) > 0$, which means that $|A_2 \cap \{g_1, g_2, g_3, g_4\}| \geq 3$.
This implies that $|A_1 \cap \{g_1, g_2, g_3, g_4\}| \leq 4-3 = 1$, and for any $g \in A_2$, it holds that $|(A_2 \setminus \{g\}) \cap \{g_1, g_2, g_3, g_4\}| \geq 3-1 = 2$.
Consequently, we have $u_1(A_1) \leq 1/4$, whereas $u_1(A_2 \setminus \{g\}) \geq 2/4$ for any $g \in A_2$.
Therefore, agent $1$ is not EF1 towards agent $2$.
It follows that any EF1 allocation cannot have a positive egalitarian welfare.

On the other hand, consider the allocation where agent $1$ receives good $g_1$, agent $2$ receives goods $g_2$, $g_3$, $g_4$, while each agent $i \in N \setminus \{1, 2\}$ receives good $g_{i+2}$.
In this allocation, agent $1$ receives a utility of $1/4$ and each of the remaining agents receives a utility of $1$.
Hence, the maximum egalitarian welfare of this instance is at least $1/4$.
Therefore, the weak egalitarian price of EF1 is $(1/4)/0 = \infty$.
\end{proof}

\subsection{Subadditive utilities}
\label{subsec:subadditive}

Since the general class of monotone utilities prove to be intractable, we next turn our attention to the class of \emph{subadditive} utilities.
We remark that most of our upper bound proofs still hold even under subadditive utilities,
the only exceptions being the upper bound proofs for the weak prices of the vanilla ECE algorithm (\cref{prop:weak_price}) and the strong egalitarian price of the heuristic choosing an unenvied agent and an unallocated good maximizing the utility when $n = 2$ (\cref{thm:agent_good_with_max_util}(b)).

For the weak prices of the vanilla ECE algorithm, the proof of \cref{prop:weak_price} does not apply to subadditive utilities because it uses the weak prices of the round-robin algorithm (\cref{lem:weak_price_rr}), which guarantees EF1 \emph{only} for additive utilities.
However, since the proof of \cref{thm:agent_good_with_max_util_first_round_only} works under subadditive utilities, it follows that the weak utilitarian price of the vanilla ECE algorithm under subadditive utilities is strictly less than $n + \sqrt{n/2} - 1/5$, and therefore is in $O(n)$.
Similarly, by \cref{thm:max_min_matching}, the weak egalitarian price of the vanilla ECE algorithm under subadditive utilities is at most $m-n+1$.

On the other hand, for the heuristic choosing an unenvied agent and an unallocated good maximizing the utility in the case where $n = 2$, we show via the following example that the strong egalitarian price under subadditive utilities is $\infty$.
This is much higher than in the additive case, where the strong egalitarian price is at most $3$ (\cref{thm:agent_good_with_max_util}(b)).

\begin{example}
Consider an instance with $n = 2$ agents and $m = 2$ goods such that agent~$1$ values any non-empty set at~$1$, while agent~$2$ values goods~$g_1$ and~$g_2$ at~$1-\epsilon$ and~$\epsilon$ respectively, where $0 < \epsilon < 1/2$. Furthermore, we assume that each agent values the empty set at $0$, and the bundle of all goods $M$ at~$1$.
In particular, note that agent~$2$'s utility is additive, while agent~$1$'s utility is non-additive.

We now consider the following allocation obtained using the ECE algorithm with the heuristic choosing an unenvied agent and an unallocated good maximizing the utility.
In the first round, the heuristic must choose agent~$1$, but can choose either good;
suppose that it chooses agent~$1$ and good~$g_1$.
Then, in the second round, the algorithm must allocate good~$g_2$ to agent~$2$.
Agent~$1$ does not envy agent~$2$, so there are no bundle exchanges involved.
Hence, the egalitarian welfare of this allocation is $\epsilon$, which is the utility of agent~$2$.

On the other hand, consider the allocation where the bundles are swapped---that is, agent~$1$ is allocated good~$g_2$, while agent~$2$ is allocated good~$g_1$.
Then, each agent receives a utility of at least $1-\epsilon$.
Therefore, the egalitarian price of this heuristic for this instance is $(1-\epsilon)/\epsilon$, which approaches $\infty$ as $\epsilon \to 0$.
\end{example}

\section{Omitted examples}
\label{ap:examples}

In the following example, we demonstrate the significance of the last envy cycle elimination step (at the end of round~$m$) of the ECE algorithm.
In particular, we show that without this step, it is possible that the utilitarian welfare of the final allocation is $0$.
By contrast, we prove in \cref{lem:no_envy_cycle} that executing the last envy cycle elimination step guarantees a utilitarian welfare of at least $1/n$.

\begin{example}
\label{ex:ec_strong_price}
Consider an instance with $n$ agents and $m = n$~goods such that for each agent $i \in N$, $u_i(g_i) = 1$ and $u_i(g_j) = 0$ if $j \neq i$.
Suppose that the ECE algorithm allocates good~$g_r$ to agent~$r+1$ in each round~$r \in \{1, \dots, n-1\}$, and allocates good~$g_n$ to agent~$1$ in round~$n$.
After each round~$r \in \{1, \dots, n-1\}$, agent~$r+1$ is envied by agent~$r$.
Furthermore, after round~$n$, agent~$1$ is envied by agent~$n$.
Hence, an envy cycle is created only after the allocation of the last good.

If we use the original version of the ECE algorithm \citep{LiptonMaMo04}, which terminates immediately after the last good is allocated, then every agent will have a utility of~$0$, resulting in a utilitarian welfare of~$0$.
On the other hand, our version of the ECE algorithm (\cref{alg:ece}) includes an additional envy cycle elimination step after the allocation of the last good.
This results in good~$g_r$ being given to agent~$r$ for each $r\in\{1,\dots,n\}$ in the final allocation.
Consequently, each agent's utility increases to~$1$, and the utilitarian welfare increases to~$n$.
\end{example}

The next example shows that the ECE algorithm is not always capable of producing an EF1 allocation with the optimal utilitarian welfare, regardless of the choices that it makes.

\begin{example}
\label{ex:ef1_not_ece}
Let $0 < \epsilon < 1/6$. Consider an instance with $3$ agents, $4$ goods, and the following utilities:

\begin{center}
\begin{tabular}{c|cccc}
    $g$      & $g_1$ & $g_2$ & $g_3$ & $g_4$ \\
    \hline
    $u_1(g)$ & $1/3$ & $0$ & $1/3+\epsilon$ & $1/3-\epsilon$ \\
    $u_2(g)$ & $0$ & $1/3$ & $1/3-\epsilon$ & $1/3+\epsilon$ \\
    $u_3(g)$ & $0$ & $0$ & $1/2$ & $1/2$
\end{tabular}
\end{center}
Note that $\mathcal{A} = (\{g_1\}, \{g_2\}, \{g_3, g_4\})$ is the unique allocation that maximizes the utilitarian welfare.
Moreover, this allocation is EF1.

Suppose, for the sake of contradiction, that allocation~$\mathcal{A}$ can be obtained through the ECE algorithm.
Consider the last good~$g$ to be added to the bundle $\{g_3, g_4\}$, and the round~$r$ in which this happens.
By definition of the ECE algorithm, the owner of the bundle~$\{g_3, g_4\} \setminus \{g\}$ must be unenvied at the beginning of round~$r$.
That is, every agent must value her current bundle at least as much as $\{g_3, g_4\} \setminus \{g\}$.

We distinguish two cases.
First, suppose that $g = g_3$.
Then, agent $2$ must have utility at least $u_2(g_4) = 1/3 + \epsilon$ at the beginning of round $r$.
Since agent $2$ ends up with a utility of $u_2(g_2) = 1/3$ at the end of the algorithm, this contradicts the fact that the ECE algorithm never decreases an agent's utility.
On the other hand, if $g = g_4$, then agent $1$ must have utility at least $u_1(g_3) = 1/3 + \epsilon$ at the beginning of round $r$, leading to a similar contradiction as agent $1$ ends up with a utility of $u_1(g_1) = 1/3$.
\end{example}

\section{Experiments (cont.)}
\label{ap:experiments}
\begin{figure*}[t]
\centering
\fontsize{6}{6}\selectfont
\tikzset{every mark/.append style={scale=0.5}}
\begin{tikzpicture}[scale=1.12]
    \begin{axis}[
        name = exp2util,
        height = 4cm, width = 4cm,
        align = center,
        title = \fontsize{7}{6}\selectfont{Utilitarian price with $2$ agents},
        xlabel = {Number of goods},
        ylabel = {Utilitarian price},
        xmin = 2, xmax = 9, ymin = 1, ymax = 1.26,
        xtick = {2, 3, 4, 5, 6, 7, 8, 9},
        minor y tick num = 4,
        legend style = {
            at = {(0cm, -0.8cm)},
            anchor = north west,
            cells = {anchor = west},
            legend columns = 3,
            legend to name = {leg:exp},
            /tikz/every even column/.append style={column sep=0.2cm},
        }
    ]
    
    \addplot[color = black, mark = *]
        coordinates {
            (2, 1.07)
            (3, 1.12)
            (4, 1.15)
            (5, 1.18)
            (6, 1.20)
            (7, 1.22)
            (8, 1.23)
            (9, 1.25)
        };
        \addlegendentry{No heuristic;}

    \addplot[color = red, mark = x]
        coordinates {
            (2, 1.00)
            (3, 1.05)
            (4, 1.06)
            (5, 1.06)
            (6, 1.06)
            (7, 1.06)
            (8, 1.06)
            (9, 1.06)
        };
        \addlegendentry{Agent with maximum utility;}
        
    \addplot[color = YellowOrange, mark = +]
        coordinates {
            (2, 1.07)
            (3, 1.07)
            (4, 1.09)
            (5, 1.09)
            (6, 1.10)
            (7, 1.09)
            (8, 1.09)
            (9, 1.09)
        };
        \addlegendentry{Good with maximum utility;}

    \addplot[color = Green, mark = o]
        coordinates {
            (2, 1.00)
            (3, 1.06)
            (4, 1.11)
            (5, 1.14)
            (6, 1.17)
            (7, 1.20)
            (8, 1.22)
            (9, 1.24)
        };
        \addlegendentry{Agent \& good with maximum utility (first round only);}
    
    \addplot[color = blue, mark = square*]
        coordinates {
            (2, 1.00)
            (3, 1.02)
            (4, 1.02)
            (5, 1.01)
            (6, 1.01)
            (7, 1.01)
            (8, 1.00)
            (9, 1.00)
        };
        \addlegendentry{Agent \& good with maximum utility;}
    
    \addplot[color = Orchid, mark = triangle*]
        coordinates {
            (2, 1.00)
            (3, 1.05)
            (4, 1.10)
            (5, 1.13)
            (6, 1.15)
            (7, 1.18)
            (8, 1.19)
            (9, 1.21)
        };
        \addlegendentry{Max-min matching}
    \end{axis}

    \begin{axis}[
        name = exp2egal,
        at = {($(exp2util.east)+(1.05cm,0)$)},
        anchor = west,
        height = 4cm, width = 4cm,
        align = center,
        title = \fontsize{7}{6}\selectfont{Egalitarian price with $2$ agents},
        xlabel = {Number of goods},
        ylabel = {Egalitarian price},
        xmin = 2, xmax = 9, ymin = 1, ymax = 1.42,
        xtick = {2, 3, 4, 5, 6, 7, 8, 9},
        minor y tick num = 4,
    ]
    
    \addplot[color = black, mark = *]
        coordinates {
            (2, 1.14)
            (3, 1.23)
            (4, 1.29)
            (5, 1.34)
            (6, 1.36)
            (7, 1.38)
            (8, 1.40)
            (9, 1.41)
        };

    \addplot[color = red, mark = x]
        coordinates {
            (2, 1.00)
            (3, 1.10)
            (4, 1.12)
            (5, 1.13)
            (6, 1.13)
            (7, 1.14)
            (8, 1.13)
            (9, 1.13)
        };
        
    \addplot[color = YellowOrange, mark = +]
        coordinates {
            (2, 1.14)
            (3, 1.10)
            (4, 1.18)
            (5, 1.25)
            (6, 1.31)
            (7, 1.34)
            (8, 1.37)
            (9, 1.39)
        };

    \addplot[color = Green, mark = o]
        coordinates {
            (2, 1.00)
            (3, 1.08)
            (4, 1.18)
            (5, 1.27)
            (6, 1.33)
            (7, 1.36)
            (8, 1.39)
            (9, 1.41)
        };
    
    \addplot[color = blue, mark = square*]
        coordinates {
            (2, 1.00)
            (3, 1.02)
            (4, 1.04)
            (5, 1.05)
            (6, 1.06)
            (7, 1.07)
            (8, 1.07)
            (9, 1.07)
        };
    
    \addplot[color = Orchid, mark = triangle*]
        coordinates {
            (2, 1.00)
            (3, 1.07)
            (4, 1.18)
            (5, 1.27)
            (6, 1.32)
            (7, 1.36)
            (8, 1.38)
            (9, 1.40)
        };
    \end{axis}

    \begin{axis}[
        name = exp3util,
        at = {($(exp2egal.east)+(1.05cm,0)$)},
        anchor = west,
        height = 4cm, width = 4cm,
        align = center,
        title = \fontsize{7}{6}\selectfont{Utilitarian price with $3$ agents},
        xlabel = {Number of goods},
        ylabel = {Utilitarian price},
        xmin = 3, xmax = 9, ymin = 1, ymax = 1.32,
        xtick = {3, 4, 5, 6, 7, 8, 9},
        minor y tick num = 4,
    ]
    
    \addplot[color = black, mark = *]
        coordinates {
            (3, 1.13)
            (4, 1.17)
            (5, 1.21)
            (6, 1.24)
            (7, 1.27)
            (8, 1.29)
            (9, 1.31)
        };

    \addplot[color = red, mark = x]
        coordinates {
            (3, 1.04)
            (4, 1.08)
            (5, 1.09)
            (6, 1.10)
            (7, 1.10)
            (8, 1.10)
            (9, 1.10)
        };
        
    \addplot[color = YellowOrange, mark = +]
        coordinates {
            (3, 1.13)
            (4, 1.11)
            (5, 1.12)
            (6, 1.12)
            (7, 1.12)
            (8, 1.12)
            (9, 1.12)
        };

    \addplot[color = Green, mark = o]
        coordinates {
            (3, 1.06)
            (4, 1.11)
            (5, 1.15)
            (6, 1.19)
            (7, 1.23)
            (8, 1.26)
            (9, 1.28)
        };
    
    \addplot[color = blue, mark = square*]
        coordinates {
            (3, 1.03)
            (4, 1.03)
            (5, 1.03)
            (6, 1.02)
            (7, 1.02)
            (8, 1.02)
            (9, 1.01)
        };
    
    \addplot[color = Orchid, mark = triangle*]
        coordinates {
            (3, 1.04)
            (4, 1.08)
            (5, 1.12)
            (6, 1.16)
            (7, 1.19)
            (8, 1.22)
            (9, 1.25)
        };
    \end{axis}
    
    \begin{axis}[
        name = exp3egal,
        at = {($(exp3util.east)+(1.05cm,0)$)},
        anchor = west,
        height = 4cm, width = 4cm,
        align = center,
        title = \fontsize{7}{6}\selectfont{Egalitarian price with $3$ agents},
        xlabel = {Number of goods},
        ylabel = {Egalitarian price},
        xmin = 3, xmax = 9, ymin = 1, ymax = 1.62,
        xtick = {3, 4, 5, 6, 7, 8, 9},
        ytick = {1, 1.1, 1.2, 1.3, 1.4, 1.5, 1.6},
        minor y tick num = 4,
    ]
    
    \addplot[color = black, mark = *]
        coordinates {
            (3, 1.27)
            (4, 1.41)
            (5, 1.48)
            (6, 1.52)
            (7, 1.55)
            (8, 1.59)
            (9, 1.61)
        };

    \addplot[color = red, mark = x]
        coordinates {
            (3, 1.08)
            (4, 1.23)
            (5, 1.27)
            (6, 1.28)
            (7, 1.29)
            (8, 1.29)
            (9, 1.29)
        };
        
    \addplot[color = YellowOrange, mark = +]
        coordinates {
            (3, 1.32)
            (4, 1.23)
            (5, 1.27)
            (6, 1.33)
            (7, 1.40)
            (8, 1.45)
            (9, 1.50)
        };

    \addplot[color = Green, mark = o]
        coordinates {
            (3, 1.17)
            (4, 1.27)
            (5, 1.34)
            (6, 1.40)
            (7, 1.45)
            (8, 1.49)
            (9, 1.54)
        };
    
    \addplot[color = blue, mark = square*]
        coordinates {
            (3, 1.17)
            (4, 1.13)
            (5, 1.13)
            (6, 1.14)
            (7, 1.16)
            (8, 1.17)
            (9, 1.17)
        };
    
    \addplot[color = Orchid, mark = triangle*]
        coordinates {
            (3, 1.00)
            (4, 1.12)
            (5, 1.24)
            (6, 1.33)
            (7, 1.42)
            (8, 1.48)
            (9, 1.53)
        };
    \end{axis}
\end{tikzpicture}
\ref{leg:exp}
\caption{Average utilitarian and egalitarian prices under the exponential distribution with parameter $\lambda = 1$}
\label{fig:exponential}
\end{figure*}
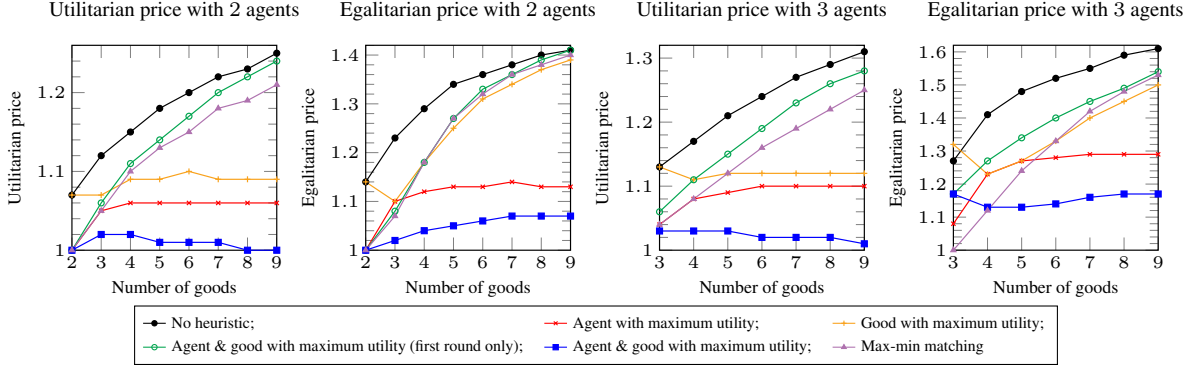

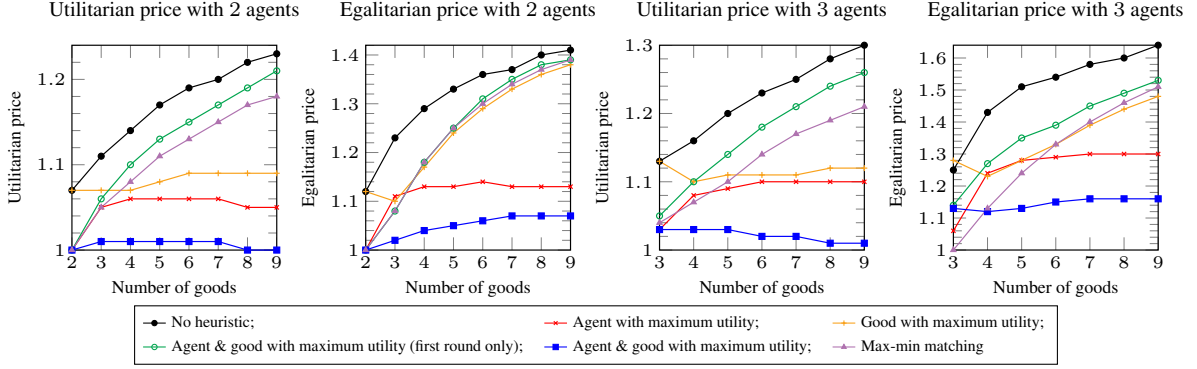
\begin{figure*}[t]
\centering
\fontsize{6}{6}\selectfont
\tikzset{every mark/.append style={scale=0.5}}
\begin{tikzpicture}[scale=1.12]
    \begin{axis}[
        name = lognorm2util,
        height = 4cm, width = 4cm,
        align = center,
        title = \fontsize{7}{6}\selectfont{Utilitarian price with $2$ agents},
        xlabel = {Number of goods},
        ylabel = {Utilitarian price},
        xmin = 2, xmax = 9, ymin = 1, ymax = 1.24,
        xtick = {2, 3, 4, 5, 6, 7, 8, 9},
        minor y tick num = 4,
        legend style = {
            at = {(0cm, -0.8cm)},
            anchor = north west,
            cells = {anchor = west},
            legend columns = 3,
            legend to name = {leg:lognorm},
            /tikz/every even column/.append style={column sep=0.2cm},
        }
    ]
    
    \addplot[color = black, mark = *]
        coordinates {
            (2, 1.07)
            (3, 1.11)
            (4, 1.14)
            (5, 1.17)
            (6, 1.19)
            (7, 1.20)
            (8, 1.22)
            (9, 1.23)
        };
        \addlegendentry{No heuristic;}

    \addplot[color = red, mark = x]
        coordinates {
            (2, 1.00)
            (3, 1.05)
            (4, 1.06)
            (5, 1.06)
            (6, 1.06)
            (7, 1.06)
            (8, 1.05)
            (9, 1.05)
        };
        \addlegendentry{Agent with maximum utility;}
        
    \addplot[color = YellowOrange, mark = +]
        coordinates {
            (2, 1.07)
            (3, 1.07)
            (4, 1.07)
            (5, 1.08)
            (6, 1.09)
            (7, 1.09)
            (8, 1.09)
            (9, 1.09)
        };
        \addlegendentry{Good with maximum utility;}

    \addplot[color = Green, mark = o]
        coordinates {
            (2, 1.00)
            (3, 1.06)
            (4, 1.10)
            (5, 1.13)
            (6, 1.15)
            (7, 1.17)
            (8, 1.19)
            (9, 1.21)
        };
        \addlegendentry{Agent \& good with maximum utility (first round only);}
    
    \addplot[color = blue, mark = square*]
        coordinates {
            (2, 1.00)
            (3, 1.01)
            (4, 1.01)
            (5, 1.01)
            (6, 1.01)
            (7, 1.01)
            (8, 1.00)
            (9, 1.00)
        };
        \addlegendentry{Agent \& good with maximum utility;}
    
    \addplot[color = Orchid, mark = triangle*]
        coordinates {
            (2, 1.00)
            (3, 1.05)
            (4, 1.08)
            (5, 1.11)
            (6, 1.13)
            (7, 1.15)
            (8, 1.17)
            (9, 1.18)
        };
        \addlegendentry{Max-min matching}
    \end{axis}

    \begin{axis}[
        name = lognorm2egal,
        at = {($(lognorm2util.east)+(1.05cm,0)$)},
        anchor = west,
        height = 4cm, width = 4cm,
        align = center,
        title = \fontsize{7}{6}\selectfont{Egalitarian price with $2$ agents},
        xlabel = {Number of goods},
        ylabel = {Egalitarian price},
        xmin = 2, xmax = 9, ymin = 1, ymax = 1.42,
        xtick = {2, 3, 4, 5, 6, 7, 8, 9},
        minor y tick num = 4,
    ]
    
    \addplot[color = black, mark = *]
        coordinates {
            (2, 1.12)
            (3, 1.23)
            (4, 1.29)
            (5, 1.33)
            (6, 1.36)
            (7, 1.37)
            (8, 1.40)
            (9, 1.41)
        };

    \addplot[color = red, mark = x]
        coordinates {
            (2, 1.00)
            (3, 1.11)
            (4, 1.13)
            (5, 1.13)
            (6, 1.14)
            (7, 1.13)
            (8, 1.13)
            (9, 1.13)
        };
        
    \addplot[color = YellowOrange, mark = +]
        coordinates {
            (2, 1.12)
            (3, 1.10)
            (4, 1.17)
            (5, 1.24)
            (6, 1.29)
            (7, 1.33)
            (8, 1.36)
            (9, 1.38)
        };

    \addplot[color = Green, mark = o]
        coordinates {
            (2, 1.00)
            (3, 1.08)
            (4, 1.18)
            (5, 1.25)
            (6, 1.31)
            (7, 1.35)
            (8, 1.38)
            (9, 1.39)
        };
    
    \addplot[color = blue, mark = square*]
        coordinates {
            (2, 1.00)
            (3, 1.02)
            (4, 1.04)
            (5, 1.05)
            (6, 1.06)
            (7, 1.07)
            (8, 1.07)
            (9, 1.07)
        };
    
    \addplot[color = Orchid, mark = triangle*]
        coordinates {
            (2, 1.00)
            (3, 1.08)
            (4, 1.18)
            (5, 1.25)
            (6, 1.30)
            (7, 1.34)
            (8, 1.37)
            (9, 1.39)
        };
    \end{axis}

    \begin{axis}[
        name = lognorm3util,
        at = {($(lognorm2egal.east)+(1.05cm,0)$)},
        anchor = west,
        height = 4cm, width = 4cm,
        align = center,
        title = \fontsize{7}{6}\selectfont{Utilitarian price with $3$ agents},
        xlabel = {Number of goods},
        ylabel = {Utilitarian price},
        xmin = 3, xmax = 9, ymin = 1, ymax = 1.3,
        xtick = {3, 4, 5, 6, 7, 8, 9},
        minor y tick num = 4,
    ]
    
    \addplot[color = black, mark = *]
        coordinates {
            (3, 1.13)
            (4, 1.16)
            (5, 1.20)
            (6, 1.23)
            (7, 1.25)
            (8, 1.28)
            (9, 1.30)
        };

    \addplot[color = red, mark = x]
        coordinates {
            (3, 1.03)
            (4, 1.08)
            (5, 1.09)
            (6, 1.10)
            (7, 1.10)
            (8, 1.10)
            (9, 1.10)
        };
        
    \addplot[color = YellowOrange, mark = +]
        coordinates {
            (3, 1.13)
            (4, 1.10)
            (5, 1.11)
            (6, 1.11)
            (7, 1.11)
            (8, 1.12)
            (9, 1.12)
        };

    \addplot[color = Green, mark = o]
        coordinates {
            (3, 1.05)
            (4, 1.10)
            (5, 1.14)
            (6, 1.18)
            (7, 1.21)
            (8, 1.24)
            (9, 1.26)
        };
    
    \addplot[color = blue, mark = square*]
        coordinates {
            (3, 1.03)
            (4, 1.03)
            (5, 1.03)
            (6, 1.02)
            (7, 1.02)
            (8, 1.01)
            (9, 1.01)
        };
    
    \addplot[color = Orchid, mark = triangle*]
        coordinates {
            (3, 1.04)
            (4, 1.07)
            (5, 1.10)
            (6, 1.14)
            (7, 1.17)
            (8, 1.19)
            (9, 1.21)
        };
    \end{axis}
    
    \begin{axis}[
        name = lognorm3egal,
        at = {($(lognorm3util.east)+(1.05cm,0)$)},
        anchor = west,
        height = 4cm, width = 4cm,
        align = center,
        title = \fontsize{7}{6}\selectfont{Egalitarian price with $3$ agents},
        xlabel = {Number of goods},
        ylabel = {Egalitarian price},
        xmin = 3, xmax = 9, ymin = 1, ymax = 1.64,
        xtick = {3, 4, 5, 6, 7, 8, 9},
        ytick = {1, 1.1, 1.2, 1.3, 1.4, 1.5, 1.6},
        minor y tick num = 4,
    ]
    
    \addplot[color = black, mark = *]
        coordinates {
            (3, 1.25)
            (4, 1.43)
            (5, 1.51)
            (6, 1.54)
            (7, 1.58)
            (8, 1.60)
            (9, 1.64)
        };

    \addplot[color = red, mark = x]
        coordinates {
            (3, 1.06)
            (4, 1.24)
            (5, 1.28)
            (6, 1.29)
            (7, 1.30)
            (8, 1.30)
            (9, 1.30)
        };
        
    \addplot[color = YellowOrange, mark = +]
        coordinates {
            (3, 1.28)
            (4, 1.23)
            (5, 1.28)
            (6, 1.33)
            (7, 1.39)
            (8, 1.44)
            (9, 1.48)
        };

    \addplot[color = Green, mark = o]
        coordinates {
            (3, 1.14)
            (4, 1.27)
            (5, 1.35)
            (6, 1.39)
            (7, 1.45)
            (8, 1.49)
            (9, 1.53)
        };
    
    \addplot[color = blue, mark = square*]
        coordinates {
            (3, 1.13)
            (4, 1.12)
            (5, 1.13)
            (6, 1.15)
            (7, 1.16)
            (8, 1.16)
            (9, 1.16)
        };
    
    \addplot[color = Orchid, mark = triangle*]
        coordinates {
            (3, 1.00)
            (4, 1.13)
            (5, 1.24)
            (6, 1.33)
            (7, 1.40)
            (8, 1.46)
            (9, 1.51)
        };
    \end{axis}
\end{tikzpicture}
\ref{leg:lognorm}
\caption{Average utilitarian and egalitarian prices under the log-normal distribution with parameters $\mu = 0$ and $\sigma = 1$}
\label{fig:lognormal}
\end{figure*}

\begin{table}[t]
\centering
\caption{Breakdown of Spliddit instances}
\begin{tabular}{ccccccccc}
    \toprule
    \multirow{2}{*}{\parbox{1.5cm}{\centering Number\\of agents}} & \multicolumn{8}{c}{{Number of goods}} \\
    & $2$ & $3$ & $4$ & $5$ & $6$ & $7$ & $8$ & $9$ \\
    \midrule
    $2$ & $103$ & $263$ & $363$ & $80$ & $278$ & $29$ & $12$ & $6$ \\
    $3$ & $-$ & $186$ & $95$ & $67$ & $2074$ & $37$ & $10$ & $16$ \\
    \bottomrule    
\end{tabular}
\label{tab:spliddit}
\end{table}

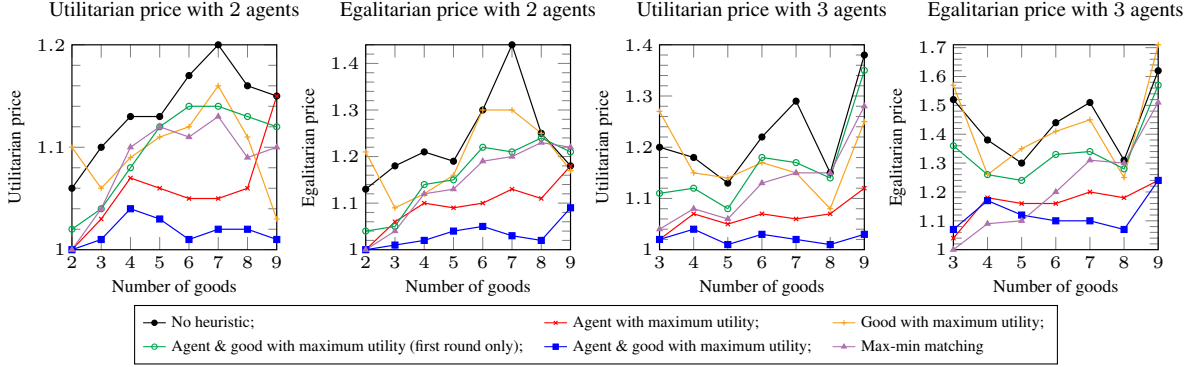
\begin{figure*}[t]
\centering
\fontsize{6}{6}\selectfont
\tikzset{every mark/.append style={scale=0.5}}
\begin{tikzpicture}[scale=1.12]
    \begin{axis}[
        name = spliddit2util,
        height = 4cm, width = 4cm,
        align = center,
        title = \fontsize{7}{6}\selectfont{Utilitarian price with $2$ agents},
        xlabel = {Number of goods},
        ylabel = {Utilitarian price},
        xmin = 2, xmax = 9, ymin = 1, ymax = 1.2,
        xtick = {2, 3, 4, 5, 6, 7, 8, 9},
        ytick = {1, 1.1, 1.2},
        minor y tick num = 4,
        legend style = {
            at = {(0cm, -0.8cm)},
            anchor = north west,
            cells = {anchor = west},
            legend columns = 3,
            legend to name = {leg:spliddit},
            /tikz/every even column/.append style={column sep=0.2cm},
        }
    ]
    
    \addplot[color = black, mark = *]
        coordinates {
            (2, 1.06)
            (3, 1.10)
            (4, 1.13)
            (5, 1.13)
            (6, 1.17)
            (7, 1.20)
            (8, 1.16)
            (9, 1.15)
        };
        \addlegendentry{No heuristic;}

    \addplot[color = red, mark = x]
        coordinates {
            (2, 1.00)
            (3, 1.03)
            (4, 1.07)
            (5, 1.06)
            (6, 1.05)
            (7, 1.05)
            (8, 1.06)
            (9, 1.15)
        };
        \addlegendentry{Agent with maximum utility;}
        
    \addplot[color = YellowOrange, mark = +]
        coordinates {
            (2, 1.10)
            (3, 1.06)
            (4, 1.09)
            (5, 1.11)
            (6, 1.12)
            (7, 1.16)
            (8, 1.11)
            (9, 1.03)
        };
        \addlegendentry{Good with maximum utility;}

    \addplot[color = Green, mark = o]
        coordinates {
            (2, 1.02)
            (3, 1.04)
            (4, 1.08)
            (5, 1.12)
            (6, 1.14)
            (7, 1.14)
            (8, 1.13)
            (9, 1.12)
        };
        \addlegendentry{Agent \& good with maximum utility (first round only);}
    
    \addplot[color = blue, mark = square*]
        coordinates {
            (2, 1.00)
            (3, 1.01)
            (4, 1.04)
            (5, 1.03)
            (6, 1.01)
            (7, 1.02)
            (8, 1.02)
            (9, 1.01)
        };
        \addlegendentry{Agent \& good with maximum utility;}
    
    \addplot[color = Orchid, mark = triangle*]
        coordinates {
            (2, 1.00)
            (3, 1.04)
            (4, 1.10)
            (5, 1.12)
            (6, 1.11)
            (7, 1.13)
            (8, 1.09)
            (9, 1.10)
        };
        \addlegendentry{Max-min matching}
    \end{axis}

    \begin{axis}[
        name = spliddit2egal,
        at = {($(spliddit2util.east)+(1.05cm,0)$)},
        anchor = west,
        height = 4cm, width = 4cm,
        align = center,
        title = \fontsize{7}{6}\selectfont{Egalitarian price with $2$ agents},
        xlabel = {Number of goods},
        ylabel = {Egalitarian price},
        xmin = 2, xmax = 9, ymin = 1, ymax = 1.44,
        xtick = {2, 3, 4, 5, 6, 7, 8, 9},
        minor y tick num = 4,
    ]
    
    \addplot[color = black, mark = *]
        coordinates {
            (2, 1.13)
            (3, 1.18)
            (4, 1.21)
            (5, 1.19)
            (6, 1.30)
            (7, 1.44)
            (8, 1.25)
            (9, 1.18)
        };

    \addplot[color = red, mark = x]
        coordinates {
            (2, 1.00)
            (3, 1.06)
            (4, 1.10)
            (5, 1.09)
            (6, 1.10)
            (7, 1.13)
            (8, 1.11)
            (9, 1.18)
        };
        
    \addplot[color = YellowOrange, mark = +]
        coordinates {
            (2, 1.21)
            (3, 1.09)
            (4, 1.12)
            (5, 1.16)
            (6, 1.30)
            (7, 1.30)
            (8, 1.25)
            (9, 1.17)
        };

    \addplot[color = Green, mark = o]
        coordinates {
            (2, 1.04)
            (3, 1.05)
            (4, 1.14)
            (5, 1.15)
            (6, 1.22)
            (7, 1.21)
            (8, 1.24)
            (9, 1.21)
        };
    
    \addplot[color = blue, mark = square*]
        coordinates {
            (2, 1.00)
            (3, 1.01)
            (4, 1.02)
            (5, 1.04)
            (6, 1.05)
            (7, 1.03)
            (8, 1.02)
            (9, 1.09)
        };
    
    \addplot[color = Orchid, mark = triangle*]
        coordinates {
            (2, 1.00)
            (3, 1.04)
            (4, 1.12)
            (5, 1.13)
            (6, 1.19)
            (7, 1.20)
            (8, 1.23)
            (9, 1.22)
        };
    \end{axis}

    \begin{axis}[
        name = spliddit3util,
        at = {($(spliddit2egal.east)+(1.05cm,0)$)},
        anchor = west,
        height = 4cm, width = 4cm,
        align = center,
        title = \fontsize{7}{6}\selectfont{Utilitarian price with $3$ agents},
        xlabel = {Number of goods},
        ylabel = {Utilitarian price},
        xmin = 3, xmax = 9, ymin = 1, ymax = 1.4,
        xtick = {3, 4, 5, 6, 7, 8, 9},
        minor y tick num = 4,
    ]
    
    \addplot[color = black, mark = *]
        coordinates {
            (3, 1.20)
            (4, 1.18)
            (5, 1.13)
            (6, 1.22)
            (7, 1.29)
            (8, 1.15)
            (9, 1.38)
        };

    \addplot[color = red, mark = x]
        coordinates {
            (3, 1.02)
            (4, 1.07)
            (5, 1.05)
            (6, 1.07)
            (7, 1.06)
            (8, 1.07)
            (9, 1.12)
        };
        
    \addplot[color = YellowOrange, mark = +]
        coordinates {
            (3, 1.27)
            (4, 1.15)
            (5, 1.14)
            (6, 1.17)
            (7, 1.15)
            (8, 1.08)
            (9, 1.25)
        };

    \addplot[color = Green, mark = o]
        coordinates {
            (3, 1.11)
            (4, 1.12)
            (5, 1.08)
            (6, 1.18)
            (7, 1.17)
            (8, 1.14)
            (9, 1.35)
        };
    
    \addplot[color = blue, mark = square*]
        coordinates {
            (3, 1.02)
            (4, 1.04)
            (5, 1.01)
            (6, 1.03)
            (7, 1.02)
            (8, 1.01)
            (9, 1.03)
        };
    
    \addplot[color = Orchid, mark = triangle*]
        coordinates {
            (3, 1.04)
            (4, 1.08)
            (5, 1.06)
            (6, 1.13)
            (7, 1.15)
            (8, 1.15)
            (9, 1.28)
        };
    \end{axis}
    
    \begin{axis}[
        name = spliddit3egal,
        at = {($(spliddit3util.east)+(1.05cm,0)$)},
        anchor = west,
        height = 4cm, width = 4cm,
        align = center,
        title = \fontsize{7}{6}\selectfont{Egalitarian price with $3$ agents},
        xlabel = {Number of goods},
        ylabel = {Egalitarian price},
        xmin = 3, xmax = 9, ymin = 1, ymax = 1.71,
        xtick = {3, 4, 5, 6, 7, 8, 9},
        ytick = {1, 1.1, 1.2, 1.3, 1.4, 1.5, 1.6, 1.7},
        minor y tick num = 4,
    ]
    
    \addplot[color = black, mark = *]
        coordinates {
            (3, 1.52)
            (4, 1.38)
            (5, 1.30)
            (6, 1.44)
            (7, 1.51)
            (8, 1.31)
            (9, 1.62)
        };

    \addplot[color = red, mark = x]
        coordinates {
            (3, 1.04)
            (4, 1.18)
            (5, 1.16)
            (6, 1.16)
            (7, 1.20)
            (8, 1.18)
            (9, 1.24)
        };
        
    \addplot[color = YellowOrange, mark = +]
        coordinates {
            (3, 1.57)
            (4, 1.26)
            (5, 1.35)
            (6, 1.41)
            (7, 1.45)
            (8, 1.25)
            (9, 1.71)
        };

    \addplot[color = Green, mark = o]
        coordinates {
            (3, 1.36)
            (4, 1.26)
            (5, 1.24)
            (6, 1.33)
            (7, 1.34)
            (8, 1.28)
            (9, 1.57)
        };
    
    \addplot[color = blue, mark = square*]
        coordinates {
            (3, 1.07)
            (4, 1.17)
            (5, 1.12)
            (6, 1.10)
            (7, 1.10)
            (8, 1.07)
            (9, 1.24)
        };
    
    \addplot[color = Orchid, mark = triangle*]
        coordinates {
            (3, 1.00)
            (4, 1.09)
            (5, 1.10)
            (6, 1.20)
            (7, 1.31)
            (8, 1.30)
            (9, 1.51)
        };
    \end{axis}
\end{tikzpicture}
\ref{leg:spliddit}
\caption{Average utilitarian and egalitarian prices of Spliddit instances}
\label{fig:spliddit}
\end{figure*}

We present experimental results that are omitted from \cref{sec:experiments}.
In particular, \Cref{fig:exponential} displays the results when the agents' utilities are sampled from the exponential distribution with parameter $\lambda = 1$, while \Cref{fig:lognormal} shows the results produced using the log-normal distribution with parameters $\mu = 0$ and $\sigma = 1$.
Besides the experiments on synthetic instances, we also conduct an experiment on real-world instances from Spliddit, a popular fair division website~\citep{GoldmanPr14}.\footnote{This dataset has been used for empirical research in various fair division papers~\citep{Shah17}.}
The dataset contains a total of $3619$ instances with $n \in \{2, 3\}$ and $m \in \{n, \dots, 9\}$.
The breakdown on the number of instances across different pairs $(n, m)$ is provided in \cref{tab:spliddit}, and the results obtained using these instances are presented in \cref{fig:spliddit}.
While the Spliddit dataset may better reflect real-world behavior, one drawback is the limited sample size, which may affect the accuracy of the results.
Nevertheless, similar patterns can be observed across all tested probability distributions as well as the Spliddit dataset.

Despite the infinite strong egalitarian prices of all proposed heuristics except the max-min matching heuristic, the experiments show that their average egalitarian prices are less than $2$ for every $n \in \{2, 3\}$ and $m \in \{n, n+1, \dots, 9\}$ (where ``average'' is defined using the harmonic mean, as explained in \cref{sec:experiments}).
This illustrates that the loss of egalitarian welfare in an allocation obtained by these heuristics in the average case is not as significant as in the worst case.

As mentioned in \cref{sec:experiments}, the heuristic that chooses an unenvied agent~$i$ and an unallocated good~$g$ with the highest utility~$u_i(g)$ in every round performs better than other heuristics.
This result generally holds even when the probability distibution is exponential or log-normal, or when instances from the Spliddit dataset are used.
However, when there are $n = 3$ agents and $m \leq 5$ goods, under non-uniform distributions, other heuristics may have a lower average egalitarian price than the heuristic that chooses an unenvied agent~$i$ and an unallocated good~$g$ with the highest utility~$u_i(g)$ in every round.

While the heuristic that chooses an agent and a good with the highest utility in the first round (resp.~the max-min matching heuristic) performs much better than the vanilla ECE algorithm in terms of its strong utilitarian price (resp.~strong egalitarian price), the improvement is not as significant in the average-case.
On the other hand, the two heuristics with similar strong utilitarian and egalitarian prices to the vanilla ECE algorithm can perform well on average.
In particular, the heuristic choosing an agent with maximum utility has much better average utilitarian and egalitarian prices than the ECE algorithm with no heuristic for each $n \in \{2, 3\}$ and $m \in \{n+1, \dots, 9\}$.\footnote{The only exception is the case $(n,m) = (2,9)$ with the Spliddit dataset. However, this might be due to the small sample size of $6$ in this case.}
Moreover, the heuristic choosing a good with the maximum utility has a much better average utilitarian price than the vanilla ECE algorithm for each $n \in \{2, 3\}$ and $m \in \{n+1, \dots, 9\}$ on synthetic data.

A possible explanation for the above observation is that both the heuristic that chooses an agent and a good maximizing the utility as well as the max-min matching heuristic optimize the welfare only in the first few rounds.
While doing so is sufficient to significantly improve the worst-case welfare guarantees, on average, the effect of such early optimization diminishes as the number of goods (and thus the number of rounds) increases.
In contrast, despite certain counterexamples where optimizing only the agents or only the goods does not lead to significant welfare improvements compared to arbitrarily choosing the agents and goods in the ECE algorithm, the effect of such optimization is far more significant in the average case.

\end{document}